\documentclass[11pt]{article}
\usepackage{a4wide}
\usepackage{authblk}
\usepackage{gensymb}
\usepackage{amsmath,amsfonts,amssymb,amsthm,graphicx,graphics,epsf}
\usepackage[ansinew]{inputenc}
\usepackage{nicefrac, wrapfig}
\usepackage{subcaption}
\usepackage[textsize=footnotesize,color=green!40]{todonotes}
\usepackage[textsize=footnotesize,color=green!40]{todonotes}
\usepackage{url}
\usepackage{color}
\usepackage{algorithm}
\usepackage[noend]{algpseudocode}
\usepackage{paralist}

\newtheorem{theorem}{Theorem}[section]
\newtheorem{corollary}[theorem]{Corollary}
\newtheorem{lemma}[theorem]{Lemma}
\newtheorem{proposition}[theorem]{Proposition}

\newtheorem{definition}[theorem]{Definition}

\newtheorem{property}[theorem]{Property}

\newcommand{\ray}[1]{\overrightarrow{#1}}
\newcommand{\disk}[1]{\disk_{#1}}

\newtheorem{observation}[theorem]{Observation}

\newcommand{\R}{\mathbb{R}}

\begin{document}
%
\title{ Affine invariant triangulations
\thanks{P.B. and P.C. were partially supported by NSERC.  \ 
P.C. was supported by CONACyT and by the Fonds de la Recherche Scientifique-FNRS under Grant n° MISU F 6001 1.\  
 R.S.\ was supported by projects PID2019-104129GB-I00/ AEI/ 10.13039/ 501100011033 and Gen. Cat. 2017-SGR-1640.
This project has received funding from the European Union's Horizon 2020 research and innovation programme under the Marie Sk\l{}odowska-Curie grant agreement No 734922.}}

%
%
\author[1]{Prosenjit Bose}
\author[2]{Pilar Cano}
\author[3]{Rodrigo I. Silveira}

\affil[1]{\small\it School of Computer Science, Carleton University, Ottawa\\

\small\tt jit@scs.carleton.ca} 
\affil[2]{\small\it D\'epartament d'Informatique, Universit\'e Libre de Bruxelles, Brussels\\ 

\small\tt pilar.cano@ulb.ac.be}
\affil[3]{\small\it Department de Matem\`atiques, Universitat Polit\`ecnica de Catalunya, Barcelona\\
 
\small\tt rodrigo.silveira@upc.edu}
%
\maketitle              
\begin{abstract}
We study affine invariant 2D triangulation methods.
That is, methods that produce the same triangulation for a point set $S$ for any (unknown) affine transformation of $S$. 
Our work is based on a method by Nielson [A characterization of an affine invariant triangulation. \emph{Geom. Mod}, 191-210. Springer, 1993] 
that uses the inverse of the covariance matrix of $S$ to define an affine invariant norm, denoted $A_{S}$, and an affine invariant triangulation, denoted ${DT}_{A_{S}}[S]$. 
We revisit the $A_{S}$-norm from a geometric perspective, and show that ${DT}_{A_{S}}[S]$ can be seen as a standard Delaunay triangulation of a transformed point set based on $S$. We prove that it retains all of its well-known properties such as being 1-tough, containing a perfect matching, and being a constant spanner of the complete geometric graph of $S$. 
We show that the $A_{S}$-norm extends to a hierarchy of related geometric structures such as the minimum spanning tree, nearest neighbor graph, Gabriel graph, relative neighborhood graph, and higher order versions of these graphs. 
In addition, we provide different affine invariant sorting methods of a point set $S$ and of the vertices of a polygon $P$  that can be combined with known algorithms to obtain other affine invariant triangulation methods of $S$ and of $P$.
\end{abstract}

\section{Introduction}
A \emph{triangulation} of a point set $S$ in the plane is a geometric graph such that its vertices are the points of $S$, its edges are line segments joining vertices and all of its faces (except possibly the exterior face) are triangles. Triangulations of point sets are of great interest in different areas such as approximation theory, computational geometry, computer aided geometric design, among others~\cite{aurenhammer2009optimal, bern1995mesh, de1987surface}. In particular, the computation of triangulations that are optimal with respect to certain criteria has been widely studied.
One of the most popular triangulations is the \emph{Delaunay triangulation} of a point set $S$, denoted  ${DT}[S]$, defined by having a triangle between any three points in $S$ if their circumcircle encloses no other point of $S$. This triangulation has the property that it maximizes the minimum of all the angles of the triangles in the triangulation. Other properties of the Delaunay triangulation include having the Euclidean minimum spanning tree as a subgraph, having an edge joining the closest pair of points in $S$, and being a constant spanner of the complete geometric graph on $S$. For a comprehensive survey of the Delaunay triangulation  see~\cite{aurenhammer2013voronoi, okabe2009spatial, shewchuk2016delaunay}. Another famous triangulation is the \emph{minimum weight triangulation}, denoted ${MWT}$, which minimizes the sum of the length of its edges. The Delaunay triangulation may fail to be a minimum weight triangulation by a factor of $\Theta (n)$ where $n$ is the size of $S$~\cite{kirkpatrick1980note}.

A \emph{geometric graph} is a graph $G$ with vertex set $V(G)$ a point set in $\R^2$ and its edges $E(G)$ are straight line segments that join two elements of $V(G)$. As an example, a Delaunay triangulation of a point set $S$ is a geometric graph with vertex set $S$. Note that a point set $S$ is also a geometric graph with vertex set $S$ and its set of edges as the empty set. Graphs $G$ and $H$ are \emph{isomorphic}, denoted $G\simeq H$, if and only if there exists a bijective function $f: V(G) \to V(H)$ such that edge $uv \in E(G)$ if and only if edge $f(u)f(v) \in E(H)$.

\begin{definition}
Let $A$ be an algorithm that takes a geometric graph as input and its output is a geometric graph. We say that $A$ is \emph{affine invariant} if and only if for any invertible affine transformation $\alpha$ and any geometric graph $G$\footnote{We assume that $G$ is a valid input for $A$.}, $A(G)$ and $A(\alpha(G))$ are isomorphic with respect to $\alpha$, i.e., $A(G)\simeq A(\alpha(G))$.
\end{definition}

In the context of triangulations, consider a triangulation algorithm $T$, which given a point set $S$ computes a triangulation  $T(S)$. We say that $T$ is affine invariant if and only if for any invertible affine transformation $\alpha$ (see Section~\ref{sec:defs} for a formal definition),  triangle $\triangle(pqr)$ is in $T(S)$ if and only if triangle $\triangle(\alpha(p)\alpha(q)\alpha(r))$ is in $T(\alpha(S))$.
Note that the transformation $\alpha$ is not known to the triangulation algorithm. Affine invariance of various geometric structures is an important property in areas such as graphics and  computer--aided geometric design.

It is easy to see that neither the Delaunay triangulation nor the minimum weight triangulation is affine invariant in general (see Figure~\ref{fig:non-affine}). 
This is because non-uniform stretching can make a point previously outside of a circumcircle to be inside, or edge lengths can increase non-uniformly.

\begin{figure}
\centering
	\includegraphics[scale=1]{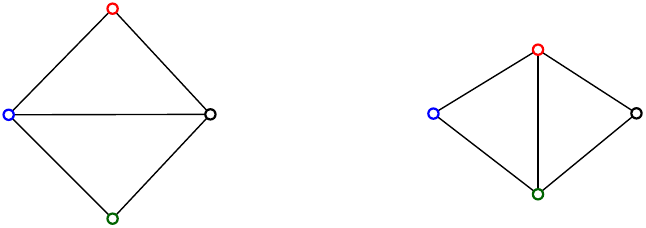}
	\caption{The points on the left correspond to an affine transformation of the points on the right, where each point is mapped to the point with the same color. The ${DT}$ and ${MWT}$ (equal in this case) differ between the left and right points, hence they are not affine invariant.} 
	\label{fig:non-affine}
\end{figure}

 Affine invariance 
is also important in the analysis and visualization of data, to guarantee for instance that different units of measurement do not influence the geometric structure, such as a triangulation that is computed. 
For this reason, Nielson~\cite{nielson1987coordinate}, in his seminal paper, defined an affine invariant normed metric  $A_{S}$ of a point set $S$, denoted $A_{S}$-norm, where for each point $v \in S$ and any affine transformation $\alpha$, $A_{S}(v)=A_{\alpha(S)}(\alpha(v))$. The $A_{S}$-norm produces ellipses (see Figure~\ref{fig:ellipses}) as the boundary of the $A_{S}$-norm disk and using this notion  Nielson~\cite{nielson1993characterization} defined an $A_{S}$-Delaunay triangulation that is affine invariant. Nielson's approach does not distinguish if the point set is rotated or reflected. While this is not an issue to obtain an affine invariant Delaunay triangulation, it makes the method unsuitable to construct other triangulations or geometric objects, like the ones discussed in Section~\ref{sec:othergeom}. 
Note that the title ``Affine invariant triangulations'' has also been used by  Haesevoets et al.~\cite{haesevoets2017affine}, who studied affine invariant methods to triangulate the union of given triangles at different times (two triangles can overlap in different regions). Thus, their objective is very different than  ours. 

\paragraph{Our work.} We revisit the $A_{S}$-norm and explain the geometry behind it in order to understand how the $A_{S}$-Delaunay triangulation behaves. We show that such triangulations have a spanning ratio related to the spanning ratio of the standard Delaunay triangulation, and that the hierarchy of subgraphs of the Delaunay triangulation, such as the minimum spanning tree~\cite{toussaint1980relative} or the relative neighborhood graph~\cite{toussaint1980relative} is also affine invariant. In addition, we describe how to use the $A_S$-norm on greedy algorithms in order to obtain other geometric objects that are affine invariant.  
Finally, we provide different algorithms to compute affine invariant orderings of a point set such as radial order, sweep-line ordering, and a polygon traversal ordering. 
Using these affine invariant orderings as subroutines, we can adapt standard geometric algorithms into affine invariant algorithms for computing a triangulation of a point set, a quadrangulation of a point set or a triangulation of a polygon.

 \begin{figure}
\centering
	\includegraphics{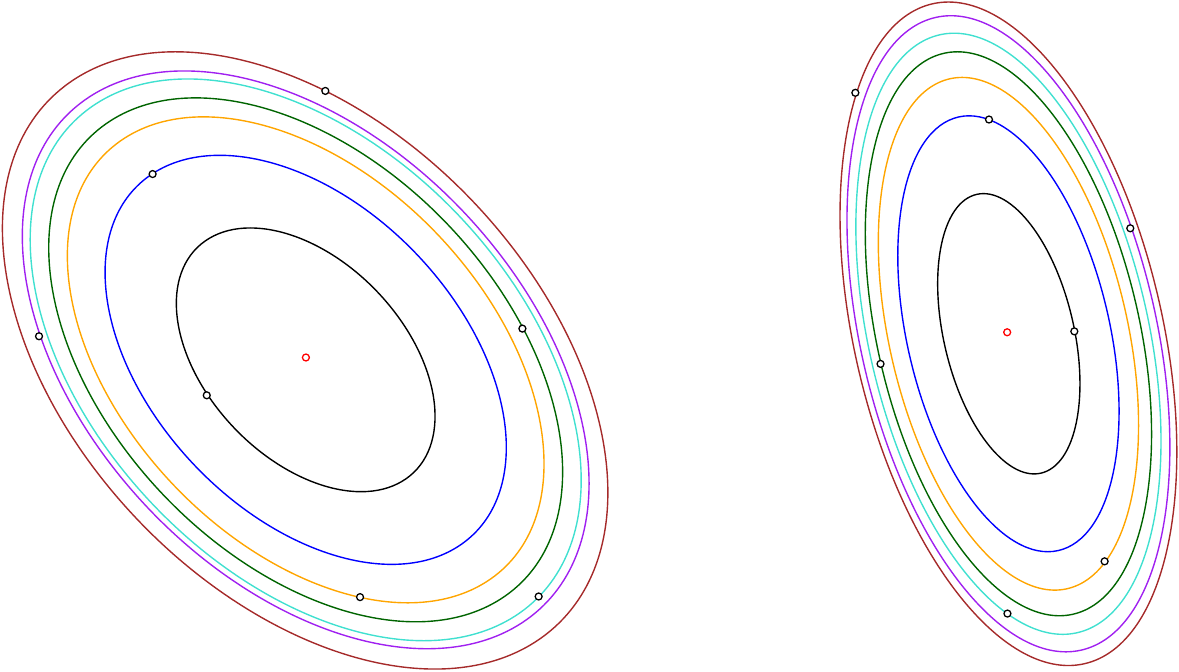}
	\caption{Point set to the right, say $S'$, is an affine transformation of the point set to the left, say $S$. Each color of the ellipses represents the corresponding boundary of the $A_S$-disk and $A_{S'}$-disk centered at the red point (mean) and containing the corresponding point of each transformation.} 
	\label{fig:ellipses}
\end{figure}

\section{Preliminaries }\label{sec:defs}

A \emph{norm} of a vector space $X$ is a nonnegative function $\rho: X \to \R^+$ with the following properties. For all $\lambda \in \R^+$ and $u,v \in X$: (a)$\rho(u+v)\leq \rho(u)+\rho(v)$, (b) $\rho(\lambda v)=\lambda \rho(v)$ and, (c) if $\rho(v)=0$ then $v$ is the zero vector. A \emph{metric} is a nonnegative function $d:X\times X \to \R^+$ such that for all $u,v,w \in X$ the following properties hold: (a) $d(u,v)=0 \iff u=v$, (b) $d(u,v)=d(v,u)$ and, (c) $d(u,w) \leq d(u,v)+d(v,w)$. When the function $d(u,v)=\rho(u-v)$ is a norm, then it is called a normed metric. Let $N$ be a normed metric, then the $N$-\emph{disk} $D_N$ centered at $c \in X$ with radius $r$ is the set of points in $X$ within $N$-distance $r$ from $c$, i.e., $D_N=\{v : v \in X$ and $N(v$$-$$c) \leq r\}$. When the radius is 1 then we call it a \emph{unit $N$-disk}. 
An \emph{affine transformation} $\alpha: X \to Y$ is of the form $\alpha(v)=Mv+b$ where $X$ is an affine space mapped to another affine space $Y$, denoted $Y = \alpha(X)$, $M$ is a linear transformation and $b$ is a vector in $\alpha(X)$. In this paper we will work in $\R^2$, i.e., $X=Y=\R^2$, $M$ is a matrix in $\R^2\times\R^2$ and $b$ is a vector in $\R^2$. For the rest of this paper we will not distinguish a point from a vector unless  notation is confusing and we will assume that $\alpha$ is invertible, i.e., it is a non-degenerate function and ${det}(M) \neq 0$. 
The following proposition states some well-known properties of affine transformations.

\begin{proposition}[\cite{byer2010methods}]
\label{prop:properties}
Let $\alpha(v)=Mv+b$ be an affine invertible transformation on $\R^2$ and let $S$ be a point set in $\R^2$. Then the function $\alpha$
\begin{compactenum}
\item\label{lines} maps lines (resp., line segments) to lines (resp., line segments), 
\item\label{parallel} preserves parallelism between lines and line segments,
\item\label{polygons} maps a simple (resp., convex) $n$-gon to a simple (resp., convex) $n$-gon,
\item\label{ratio} preserves the ratio of lengths of two parallel segments,
\item\label{areas} preserves the ratio of areas of two object, and
\item\label{mean} maps the mean of $S$ to the mean of $\alpha(S)$.
\end{compactenum}
\end{proposition}

Let $S=\{p_1, p_2, \ldots, p_n\}$ be an $n$-point set in the plane where the coordinates of each point $p_i \in S$ are denoted by $(x_i, y_i)$.  
 
 Nielson~\cite{nielson1987coordinate} defines an affine invariant normed metric, that we call $A_{S}$-norm, in the following way. 

Let 
\begin{equation*}
\mu_x= \frac{1}{n}\displaystyle\sum_{i=1}^{n} x_i,  \quad \mu_y= \frac{1}{n}\displaystyle\sum_{i=1}^{n} y_i, \quad \sigma^2_x = \frac{1}{n}\displaystyle\sum_{i=1}^{n}(x_i-\mu_x)^2 \end{equation*}
\begin{equation*} \sigma^2_y = \frac{1}{n}\displaystyle\sum_{i=1}^{n}(y_i-\mu_y)^2, \quad
\sigma_{xy}= \frac{1}{n}\displaystyle\sum_{i=1}^{n}(x_i-\mu_x)(y_i-\mu_y). \end{equation*}
Note that the mean $\mu=(\mu_x,\mu_y)$ is the barycenter of $S$. The \emph{covariance matrix} of a point set $S$ is defined as $\Sigma=\begin{pmatrix}
  \sigma_x^2 & \sigma_{xy} \\
  \sigma_{xy} & \sigma_y^2
 \end{pmatrix}$ and its inverse
$\Sigma^{-1}$$=$$\frac{1}{\sigma^2_x\sigma_y^2-\sigma_{xy}^2}$$\cdot$$\begin{pmatrix}
  \sigma_y^2 & -\sigma_{xy} \\
  -\sigma_{xy} & \sigma_x^2
 \end{pmatrix}$.

Let $p=(x,y) \in S$, then the $A_{S}$-norm metric  is defined as 
\begin{equation*}A_{S}(x,y)=(x \quad y)\Sigma^{-1}\begin{pmatrix}
 x\\
 y
 \end{pmatrix} = \frac{x^2\sigma_y^2-2xy\sigma_{xy}+y^2\sigma_x^2}{\sigma_y^2\sigma_x^2-\sigma^2_{xy}}. \end{equation*}

The matrix $\Sigma^{-1}$ is also known as the \emph{concentration matrix}~\cite{dodge2006oxford}, which defines a norm with respect to the normal (Gaussian) distribution defined by $S$. The eigenvectors of $\Sigma$ and $\Sigma^{-1}$ are the same and these vectors define the principal orthogonal directions of how spread the point set is with respect to its mean (barycenter) $\mu$. In other words, if we compute the Gaussian manifold  defined by the bivariate normal distribution given by the point set $S$  and then cut the Gaussian manifold with a plane parallel to the plane $z=0$, then we obtain an ellipse.  See Figure~\ref{fig:ellipses}. Such an ellipse has principal orthogonal axes defined by the eigenvectors of $\Sigma^{-1}$. Thus, the boundary of an $A_{S}$-disk will be defined by a homothet\footnote{A \emph{homothet} of an ellipse $D$  is obtained by uniformly scaling $D$ with respect to its center, followed by a translation.} of the resulting ellipse where the boundary of the \emph{unit $A_{S}$-disk} will be represented by the ellipse with principal axes being parallel to the eigenvectors of $\Sigma^{-1}$ and the magnitude of each principal axis will be given by the square root of the eigenvalue of the corresponding unit eigenvector. 

 The \emph{$N$-Delaunay triangulation} of $S$, denoted ${DT}_{N}[S]$, is defined in the following way. For every three distinct points $p_i, p_j, p_k$ in $S$, the triangle $\bigtriangleup(p_ip_jp_k)$ is in ${DT}_{N}[S]$  
if and only if there exists an $N$-disk containing the three points on its boundary and no other point of $S$ in its interior. 
The $L_2$-Delaunay triangulation is the standard Delaunay triangulation, simply denoted ${DT}[S]$. 

We say that a point  set $S$ is in \emph{general position} if no three points are collinear and all points in $S$ are at different $A_S$-norm distance from the mean $\mu$. 
Since  
the boundary of the $A_{S}$-disk is an ellipse, Nielson computes the $A_{S}$-Delaunay triangulation using this $A_S$-disk and shows the following.

\begin{theorem}[Nielson~\cite{nielson1993characterization}]~\label{Nielson-triang}
Let $S$ be a point set in general position and $\alpha$ an invertible affine transformation, then $DT_{A_S}[S]\simeq DT_{A_{\alpha(S)}}[\alpha(S)]$ under the affine transformation $\alpha$, i.e., the $A_{S}$-Delaunay triangulation is affine invariant.
\end{theorem}

Finally, the \emph{order type} of a point set $S$ is a mapping that assigns to each ordered triple $i,j,k \in \{1, \ldots, n\}$ the orientation of $p_i, p_j$ and $p_k$ (either clockwise or counter-clockwise). It can be shown that order types are preserved, up to a change of sign, by checking for each triple $p_i, p_j, p_k \in S$ the signed area of the triangles $\triangle(p_ip_jp_k)$ and $\triangle(\alpha(p_i)\alpha(p_j)\alpha(p_k))$ given by the following cross product $\alpha(p_i-p_j)\times\alpha(p_k-p_j)= {det}(M)((p_i-p_j)\times(p_k-p_j))$.

\section{The $A_S$-Delaunay triangulation revisited}\label{sec:affine-del}

In this section we discuss the connection between the standard Delaunay triangulation and the $A_{S}$-Delaunay triangulation. 

Let $S$ be a point set in $\R^2$ in general position. Consider the $2\times n$ matrix $N$ such that for each point $p$ in $S$ there is one column in $N$ represented by the vector $p-\mu$. 
Then, we see that $\Sigma=\frac{1}{n}{NN}^T$. If a point set $S'=\alpha(S)$ and $\alpha(p)=Mp+b$, with $p \in S$, is an affine transformation of the point set $S$, then its mean is $\alpha(\mu)$ and the covariance matrix $\Sigma'$ of $S'$ is $\Sigma'=M\Sigma M^T$. Recall that $M$ is a $2\times2$ matrix in $\R^2$ with ${det}(M)\neq 0$.
 
  Since $S$ is in general position, we have that ${det}(\Sigma) \neq 0$. Thus, $\Sigma$ is invertible. Moreover, since $\Sigma$ is a square symmetric matrix, $\Sigma$ can be represented as $Q\Lambda Q^{T}$ where $Q$ is the matrix of orthonormal eigenvectors of $\frac{1}{n}{NN}^T$, $\Lambda$ is the diagonal matrix of eigenvalues and $Q^{-1}=Q^T$. Recall that if a matrix $M$ is diagonal, then $M=M^T$ and that $M^k$ with $k \in \R^+$ has entries $M_{i,j}^k$. Therefore, we can also rewrite the covariance matrix as $(Q\Lambda^{\nicefrac{1}{2}})(Q\Lambda^{\nicefrac{1}{2}})^{T}$. Looking carefully at this representation of $\Sigma$ and $\Sigma'$, we obtain that $(Q \Lambda^{\nicefrac{1}{2}})^{-1}S$ is an affine transformation of $S$ with $\mathbb{I}$ as its covariance matrix. We refer to the point set $(Q \Lambda^{\nicefrac{1}{2}})^{-1}S$ as the point set $S$ \emph{normalized}. 
  Note that the unit $A_{(Q \Lambda^{\nicefrac{1}{2}})^{-1}S}$-disk  
is the Euclidean unit disk. This implies that the $A_{(Q\Lambda^{\nicefrac{1}{2}})^{-1}S}$-Delaunay triangulation of $(Q\Lambda^{\nicefrac{1}{2}})^{-1}S$ is the $L_2$-Delaunay triangulation of $(Q\Lambda^{\nicefrac{1}{2}})^{-1}S$, which together with Theorem~\ref{Nielson-triang} proves the following proposition.

  \begin{proposition}\label{lem:affine-inv}
Let $S$ be a point set in general position in $\R^2$ and let $\Sigma=Q\Lambda Q^T$ be its covariance matrix where $Q$ is the matrix of orthonormal eigenvectors of $\frac{1}{n}{NN}^T$ and $\Lambda$ is the diagonal matrix of eigenvalues. Then, $DT[(Q\Lambda^{\nicefrac{1}{2}})^{-1}S]$ and $DT_{A_S}[S]$ are isomorphic under the linear transformation $(Q\Lambda^{\nicefrac{1}{2}})$.
  \end{proposition}
  
   \begin{figure}
\centering
	\includegraphics[scale = 1]{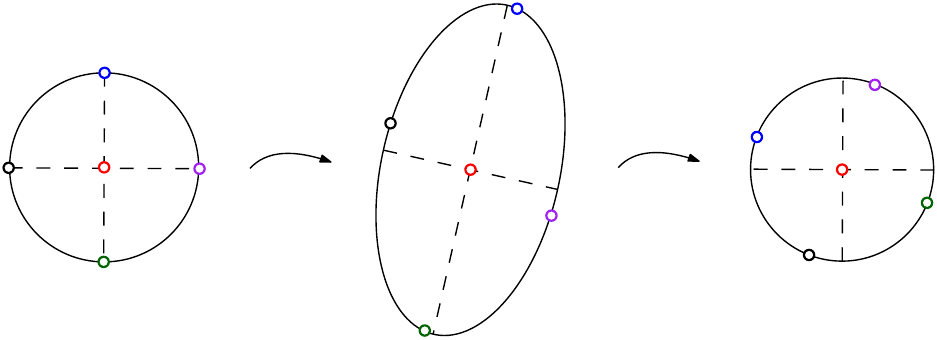}
	\caption{(\emph{left}) Point set with $A_S$-disk as the Euclidean disk, (\emph{middle}) an affine transformation of point set in the left, (\emph{right}) a resulting point set when normalizing the point set in the middle.} 
	\label{fig:trans}
\end{figure}

Let $S'=\alpha(S)$ be an affine transformation of the point set $S$ with covariance matrix $\Sigma'=Q'\Lambda' Q'^{T}$. Notice that even though the point sets $(Q\Lambda^{\nicefrac{1}{2}})^{-1}S$ and $(Q'\Lambda'^{\nicefrac{1}{2}})^{-1}S'$ have the Euclidean metric as their corresponding Nielson's norm, both point sets can be different. For instance, in Figure~\ref{fig:trans} we have on the left a point set that defines the Euclidean metric as its $A_S$-norm, in the middle we have the point set resulted from a rotation and scaling of the point set on the left. Finally, the point set on the right is a normalization of the point set in the middle. Note that the two point sets on the right and left are different by a rotation. In other words, $(Q\Lambda^{\nicefrac{1}{2}})^{-1}S$ and $(Q'\Lambda'^{\nicefrac{1}{2}})^{-1}S'$ are the same point set \emph{up to rotations and reflections}.  

 \begin{figure}
\centering
	\includegraphics[width=\textwidth]{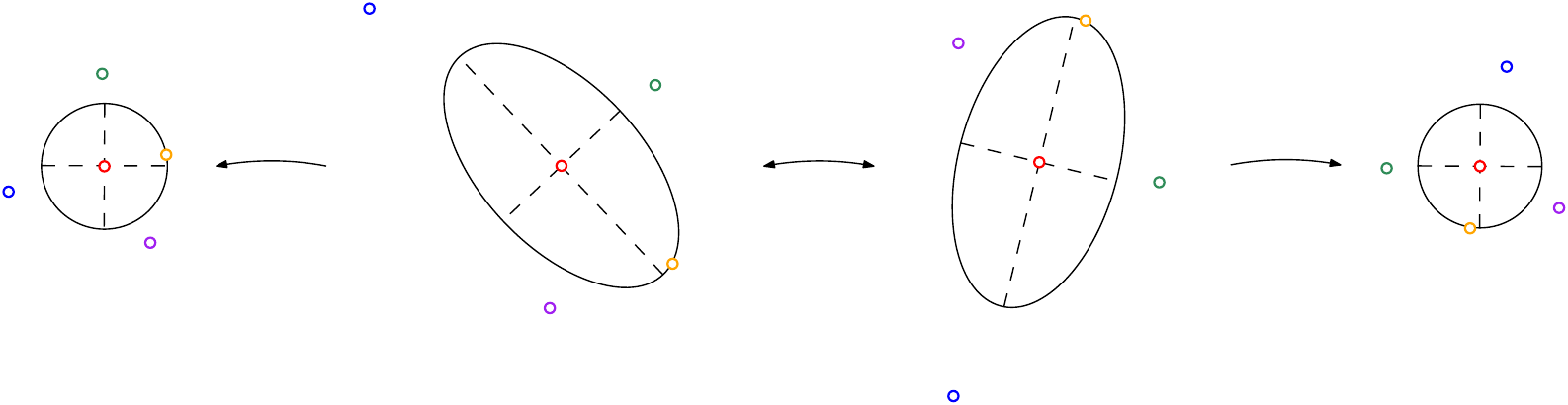}
	\caption{The two point sets in the middle are the initial point sets, one is an affine transformation (a rotation and reflection) of the other one. The point set on the left is the obtained point set when normalizing the left-middle point set and the point set on the right is the obtained point set when normalizing the right-middle point set.} 
	\label{fig:trans2}
\end{figure}

Another example is depicted in Figure~\ref{fig:trans2}, where we have two point sets in the middle that can be transformed into each other by an affine transformation -- the one in the left-middle is a reflection of the one in the right-middle by the $x$-axis and a rotation of $300\degree$. When we normalize each of the two middle point sets, we obtain the point set on the extreme left and right, respectively. Note that these two point sets are different. Moreover, the point sets on the left and right have different sign order type.

On the other hand, a nice implication of Proposition~\ref{lem:affine-inv} is that the $A_{S}$-Delaunay triangulation behaves in many ways like a standard Delaunay triangulation. For instance, the $DT_{A_S}[S]$ is $1$-tough~\cite{dillencourt1990toughness}, which implies that it contains a perfect matching when $|S|$ is even. 
 
 Given a weighted graph $G=(V,E)$ and a real number $t \geq 1$, a $t$-\emph{spanner} of $G$ is a spanning subgraph $G'$ such that for every edge $uv$ in $G$, there exists a path from $u$ to $v$ in $G'$ whose weight is no more than $t$ times the weight of the edge $uv$ in $G$. 
 The smallest value of $t$ for which this property holds is called the \emph{spanning ratio} of $G'$, denoted as $sr(G')$. We simply say that $G'$ is a spanner if $sr(G')$ is a constant.
  It is known that the standard Delaunay triangulation is a spanner, see~\cite{dobkin1990delaunay, keil1989delaunay, xia}. The following theorem shows that $DT_{A_S}[S]$ is also a spanner whose spanning ratio is a constant factor times the spanning ratio of the standard Delaunay triangulation. 
The constant factor depends on the eigenvalues of the co-variance matrix.  
 
 \begin{theorem}
 Let $S$ be a point set in general position and let $\Sigma=Q\Lambda Q^T$ be the covariance matrix of $S$. Let $\lambda_{\max}$ and $\lambda_{\min}$ be the maximum and minimum eigenvalues of $\frac{1}{n}{NN}^T$, respectively. Then, \begin{equation*}sr(DT_{A_S}[S]) \leq\Big( \frac{\lambda_{\max}}{\lambda_{\min}} \Big)^{\nicefrac{1}{2}}\cdot sr({DT}[(Q \Lambda^{\nicefrac{1}{2}})^{-1}S])\leq 1.998\Big(\frac{\lambda_{\max}}{\lambda_{\min}} \Big)^{\nicefrac{1}{2}}.\end{equation*}
 \end{theorem}
 
 \begin{proof}
 Let $S'=(Q \Lambda^{\nicefrac{1}{2}})^{-1}S$.  The triangulation $DT[S']$ is a standard Delaunay triangulation. For every pair of points $u, v \in S'$ let $\delta_{uv}$ be a shortest path from $u$ to $v$ contained in $DT[S']$. By definition the spanning ratio $\frac{\sum_{(p_i,p_j) \in \delta_{uv}}d(p_i,p_j)}{d(u,v)} \leq sr({DT}[S'])$. 
 
 Note  
that the only thing that changes the spanning ratio is when the graph $DT[S']$ is stretched with different scaling factors in the $x$- and $y$-coordinates. Such scaling is defined by the square root of the eigenvalues of $\Sigma$ given in the diagonal matrix $\Lambda$, since $Q$ represents a matrix composed by rotations and reflections and $\Lambda^{\nicefrac{1}{2}}$ represents the scaling factor of the point set $S'$ where its covariance matrix is the identity matrix $\mathbb{I}$. 
Hence, for any $u, v \in S'$ we have  $d(\Lambda^{\nicefrac{1}{2}}u,\Lambda^{\nicefrac{1}{2}}v)$ $\leq$ $\lambda_{\max}^{\nicefrac{1}{2}} \cdot d(u,v)$ and $d(\Lambda^{\nicefrac{1}{2}}u,\Lambda^{\nicefrac{1}{2}}v) \geq \lambda_{\min}^{\nicefrac{1}{2}} \cdot d(u,v)$.
 Therefore, \begin{align*}sr({DT}_{A_S}[S])
 \leq \Big(\frac{\lambda_{\max}}{\lambda_{\min}} \Big)^{\nicefrac{1}{2}}\cdot sr({DT}[S'])\leq 1.998\Big(\frac{\lambda_{\max}}{\lambda_{\min}} \Big)^{\nicefrac{1}{2}}, \end{align*}
 where the second inequality follows from Xia's result~\cite{xia}.
 \end{proof}


\section{Primitives for other affine invariant geometric constructions}\label{sec:othergeom}

A natural line of study to follow is to consider other geometric objects with affine invariant construction algorithms, such as algorithms for triangulating a point set besides the $A_S$-Delaunay triangulation, triangulating a simple polygon, or computing a $k$-angulation of a point set, among others. In this section we identify some necessary ingredients for defining such methods.

At first sight, when such constructions rely on a metric, then the $A_S$-norm can be used, for instance to compute an affine invariant Delaunay triangulation or an affine invariant minimum weight triangulation. In general, the algorithms to compute such geometric objects fall into two categories: (1) based on an empty region property, and (2) based on the rank of the length of the $n\choose{2}$ edges of a point set. Details can be found in Sections~\ref{subsec:disk} and~\ref{subsec:norm} together with different examples of algorithms that become affine invariant when the $A_S$-norm is used.

However, this does not always work. 
Many algorithmic techniques rely on the given order of the points, such as a radial order or the order obtained by sweeping the point set in a given direction. 
Yet, the use of the $A_S$-norm does not handle rotations or reflections,  
it only solves the scaling factors on the $x$- and $y$-coordinates. 
For instance, if $S$ is a point set such that the $A_S$-norm is the Euclidean distance and $S'$ is a rotation of $S$, then the boundary of the $A_{S'}$-disk is also defined by a circle, and there is no transformation needed. However, the order of a sweep-line with respect to the $x$-coordinates for $S$ and $S'$ will not necessarily be the same. Another simple example is when $P$ is a simple polygon with vertex set $S$, such that $S$ has as $A_S$-disk the Euclidean disk. Consider point set $S'$ obtained by reflecting $P$ with respect to the $x$-axis, and the resulting polygon $P'$. Again, $S'$ will have as $A_{S'}$-disk the Euclidean disk. However, the clockwise order of $S$ and $S'$ are different.

Thus, the next question is, what do we need in order to create affine invariant sorting methods? For any sorting method, we need to define which point is the initial point, so, we have to be able to choose a point $\wp$ in the point set $S$, such that for any invertible affine transformation $\alpha$, we always choose $\alpha(\wp)$ as the initial point. Second, for the cases of radial order and sweep-line we need a ray $\ray{\wp \mu}$ and a line $\ell_{\wp}$, respectively, such that for any invertible affine transformation $\alpha$, we always choose $\ray{\alpha(\wp)\alpha(\mu)}$ and line $\alpha(\ell_{\wp})$, respectively. Finally, since there might be reflections, which change clockwise for counter-clockwise orientation and left for right of a ray, we need to be able to choose the correct orientation (clockwise or counter-clockwise, right or left) whether there is a reflection or not. We can solve this problem if we are able to choose a point $\delta\neq \wp$, such that for any invertible affine transformation $\alpha$, we always choose $\alpha(\delta)$. Then the direction is given depending on whether $\alpha(\delta)$ is to the left or to the right of $\ray{\alpha(\wp)\alpha(\mu)}$. In essence, we need to find a way to compute three affine invariant points. We will see two different methods to do this.

A function $f$ of $S$ is \emph{affine invariant} if and only if $\alpha(f(S))=f(\alpha(S))$ for any invertible affine transformation $\alpha$. For instance, the mean $f(S)=\mu$ is an affine invariant function. Hence, note that if we define affine invariant functions that compute a point, a ray, and an oriented line, 
we can design affine invariant sorting methods or other algorithms that have these as primitives. 

Let $S$ be a point set in general position. Let $u, v$ be two points in $S$. We say that point $u$ \emph{is to the right of $\overrightarrow{wz}$} if the signed area of $\triangle(zwu)$ is positive. Otherwise, $u$ \emph{is to the left of $\overrightarrow{wz}$}. Let $\alpha(x)=Mx+b$ be an invertible affine transformation of $\R^2$. Using the fact that $\alpha(w-z)\times\alpha(w-u)= {det}(M)((w-z)\times(w-u))$ the sign of ${det}(M)$ can be used to determine the orientation of $\alpha(u)$ with respect to $\overrightarrow{\alpha(w)\alpha(z)}$. We say that $u$ and $v$ \emph{lie on the same side of a directed line segment $\overrightarrow{wz}$} if ${sign}((w-z)\times(w-u))={sign}((w-z)\times(w-v))$.

\begin{observation}\label{obs:direct}
Let $\overrightarrow{wz}$ be a directed line segment in $\R^2$ and let $u, v$ be two points in $S$. Let $\alpha$ be an invertible affine transformation. The points $u$ and $v$ are on the same side of $\overrightarrow{wz}$ if and only if $\alpha(u)$ and $\alpha(v)$ are on the same side of $\overrightarrow{\alpha(w)\alpha(z)}$. 
\end{observation}
\begin{proof}
 Since $u$ and $v$ are on the same side of $\overrightarrow{wz}$, ${sign}((w-z)\times(w-u))={sign}((w-z)\times(w-v))$. 
 Hence, ${sign}(\alpha(w)-\alpha(z)) \times (\alpha(w)-\alpha(u))) = {sign}({det}(M)){sign}((w-z)\times(w-u))=$  
${sign}({det}(M)){sign}((w-z)\times(w-v)) = {sign}((\alpha(w)-\alpha(z))\times(\alpha(w)-\alpha(v)))$.
\end{proof}

Note that there may be several ways of implementing affine invariant functions that compute the desired primitives for sorting, namely: a point $\wp \in S$, a ray $\overrightarrow{\wp\mu}$ and a point $\delta$ that defines an orientation with respect to $\overrightarrow{\wp\mu}$. In the rest of this section we provide two different procedures for defining these desired primitives.

In the following observation, we use the convex hull and the barycenter of $S$ in order to compute the primitives. See Figure~\ref{fig:obs4-2}.

\begin{observation}\label{obs:otherpoints}
Let $S$ be a point set in general position in the plane. Let $\mu$ be the mean of $S$ and let ${CH}(S)$ be the convex hull of $S$. 
For each edge $p_ip_{i+1} \in {CH}(S)$, consider the triangle $p_ip_{i+1}\mu$.  
Assume that the areas of these triangles are pairwise different. Let $\vartheta$ and $\delta$ be the barycenters of the triangles with largest and second largest area, respectively, denoted $\triangle_{B_1}$ and $\triangle_{B_2}$, respectively. Consider the ray $\ray{\mu\vartheta}$. If $\delta$ is to the left of $\ray{\mu\vartheta}$, then let $\wp$ be the vertex of $\triangle_{B_1}$ that is to the left of $\ray{\mu\vartheta}$. Otherwise, let $\wp$ be the vertex of $\triangle_{B_1}$ that is to the right of $\ray{\mu\vartheta}$. The functions $f_1(S)=\vartheta, f_2(S)=\delta$ and $f_3(S)=\wp$ are affine invariant. 
\end{observation}

\begin{figure}[tb]
\centering
	\includegraphics[scale=1]{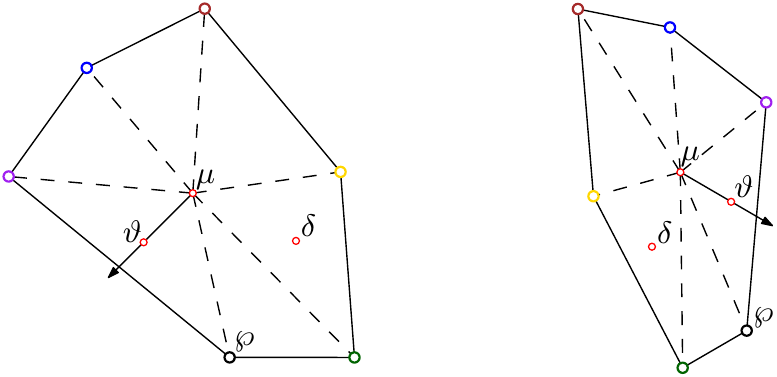}
	\caption{The bold polygons define the convex hull of two point sets that can be transform into each other by an affine transformation. Equal color indicate pairs of points and their affine transformations. The triangle defined by the points $\mu$, black and purple is the triangle with largest area. In addition, the triangle defined  by the points $\mu$, yellow and green is the triangle with  second largest area.}
	\label{fig:obs4-2}
\end{figure}

\begin{proof}
Let $S'=\alpha(S)$ be an affine transformation of $S$. Consider the convex hull ${CH}(S')$ of $S'$ and for each edge $p'_ip'_{i+1}$ of ${CH}(S')$ we consider the triangles $\triangle(p'_ip'_{i+1}\alpha(\mu))$. Let $\triangle'_{B_1}$ and $\triangle'_{B_2}$ be the triangles with largest and second largest area of the triangles defined by the edges of ${CH}(S')$ and $\alpha(\mu)$.
Let $f_1(S')=\vartheta'$ and $f_2(S')=\delta'$ be the barycenters of $\triangle'_{B_1}$ and $\triangle'_{B_2}$, respectively. If $\delta'$ is to the left of $\ray{\alpha(\mu)\vartheta}$, then let $f_3(S')=\wp'$ be the vertex in $\triangle'_{B_1}$ that is to the left of $\ray{\alpha(\mu)\vartheta'}$. Otherwise, let $f_3(S')=\wp'$ be the vertex $\triangle'_{B_1}$ that is to the right of $\ray{\alpha(\mu)\vartheta'}$.

Let us show that $\alpha(f_1(S))=\alpha(\vartheta)=\vartheta'=f_1(S'), \alpha(f_2(S))=\alpha(\delta)=\delta'=f_2(S')$ and $\alpha(f_3(S))=\alpha(\wp)=\wp'=f_3(S')$. From Proposition~\ref{prop:properties}(\ref{mean}) the mean of $S'$ is $\alpha(\mu)$. From Proposition~\ref{prop:properties}(\ref{polygons}) the convex hull of $S$ is mapped to the convex hull $S'$. Also, from Proposition~\ref{prop:properties}(\ref{polygons}), for each edge $p_ip_{i+1} \in {CH}(S)$, the triangle $\triangle(p_ip_{i+1}\mu)$ is mapped to $\triangle(\alpha(p_i)\alpha(p_{i+1})\alpha(\mu))$. 
Let $\triangle$ be a triangle defined by the endpoints of an edge in ${CH}(S)$ and $\mu$, that is different from $\triangle_{B_1}$. 
Since ${Area}(\triangle_{B_1}) > {Area}(\triangle)$, it follows that $\frac{{Area}(\triangle_{B_1})}{{Area}(\triangle)}>1$. From Proposition~\ref{prop:properties}(\ref{areas}) we have that $\frac{{Area}(\alpha(\triangle_{B_1}))}{{Area}(\alpha(\triangle))} = \frac{{Area}(\triangle_{B_1})}{{Area}(\triangle)}>1$. Thus, ${Area}(\alpha(\triangle_{B_1})) > {Area}(\alpha(\triangle))$. Therefore, $\alpha(\triangle_{B_1})=\triangle'_{B_1}$. By the same arguments, $\alpha(\triangle_{B_2})=\triangle'_{B_2}$. Hence, $\vartheta'$ and $\delta'$ are the barycenters of $\alpha(\triangle_{B_1})$ and $\alpha(\triangle_{B_2})$, respectively. Proposition~\ref{prop:properties}(\ref{mean}) implies that $\alpha(\vartheta)=\vartheta'$ and $\alpha(\delta)=\delta'$. By Proposition~\ref{prop:properties}(\ref{lines}) we have that $\ray{\mu\vartheta}$ is mapped to $\ray{\alpha(\mu)\vartheta'}$. From Observation~\ref{obs:direct} it follows that $\alpha(\wp)=\wp'$.
\end{proof}

Under the assumption that the largest and second largest area of the triangles defined by each edge of the ${CH}(S)$ and the mean $\mu$ of $S$ are different, we can define affine invariant functions that compute a point $\wp$ of $S$, a ray $\ray{\wp\mu}$, a line defined by $\mu$ and any of $\wp$ and $\vartheta$; and a direction defined by a ray $\ray{\wp\mu}$ and point $\delta$.  On the other hand, notice that using the area of a triangle formed by $\mu$ and any pair of points, we obtain an affine invariant order of the edges of the complete graph of $S$ such that each edge has weight equal to the area of the triangle formed by its endpoints and $\mu$. Under the assumption that pairwise triangles have different area, we can also obtain affine invariant algorithms that are based on the rank of the edges of the complete graph (see Section~\ref{subsec:norm}). 

We will say that a point $u$ is the \emph{$A_S$-closest} point to $\mu$ if $u$ minimizes the $A_S$-distance to $\mu$. Using the fact that the $A_S$-norm is affine invariant, we observe that there is another method for distinguishing the desired primitives for affine invariant radial and sweep-line ordered.

\begin{observation}\label{obs:closest}
Let $S$ be a point set in general position and let $\mu$ be the mean of $S$ and let $\alpha(S)$ be an affine transformation of $S$ with mean $\mu'$. The $k$-th $A_S$-closest point to $\mu$ is mapped by $\alpha$ to the $k$-th $A_{\alpha(S)}$-closest point to $\mu'$.
\end{observation}
\begin{proof}
 Since $S$ is in general position, all of the points in $S$ are at different $A_S$-distance to $\mu$. Hence, we can order the points of $S$ in increasing order with respect to the $A_S$-distance from each point to $\mu$. Since $A_S$ is an affine invariant norm, the points in $\alpha(S)$ are also at different $A_{\alpha(S)}$-distance to the mean $\alpha(\mu)=\mu'$. This again defines an order of $\alpha(S)$. Since the $A_S$-distance is invariant under affine transformations, it follows that the increasing order with respect to the $A_S$-distance from each point to $\mu$ is affine invariant. Thus, the $k$-th $A_S$-closest point to $\mu$ is mapped by $\alpha$ to the $k$-th $A_{\alpha(S)}$-closest point to $\mu'$.
\end{proof}

From Observation~\ref{obs:closest} it follows that we can also use the $A_S$-norm in order to define affine invariant functions $f_1(S)$ and $f_2(S)$, such that each function computes a point in the following fashion. Let $\mu$ be the mean of $S$, $f_1(S)=\wp$ be the $A_S$-closest point to $\mu$. Let $f_2(S)=\delta$ be the second $A_S$-closest point to $\mu$ if it is not on the line defined by $\mu$ and $\wp$. Otherwise, let $f_2(S)=\delta$ be the third $A_S$-closest point to $\mu$. 

Therefore, from Observations~\ref{obs:otherpoints} and~\ref{obs:closest} we obtain the following theorem.

\begin{theorem}\label{thm:affine_functions}
There exist three affine invariant functions $f_1, f_2, f_3$ for any point set $S$ in general position, such that $f_1(S), f_2(S)$ and $f_3(S)$ are three different points in the plane that are non collinear and at least one of them is in $S$.   
\end{theorem}

Finally, we obtain the next as a corollary of Theorem~\ref{thm:affine_functions}.

\begin{corollary}
Let $S$ be a point set in general position. Any  geometric algorithm based on a point in $S$ and a direction or orientation can be made affine invariant. 
\end{corollary}

In the following section we present two algorithms that are based on a point in $S$ and a direction and orientation that can be made affine invariant.

\section{Affine invariant sorting algorithms of a point set}\label{sec:sorting}

In this section we present two different sorting algorithms that together with the affine invariant primitives defined in Section~\ref{sec:othergeom} result in two different affine invariant sorting methods. 

\subsection{Affine invariant radial ordering}\label{sec:radial}

Let $\mu, \delta$ and $\wp$ be three non-collinear points, such that $\wp$ is in $S$. Consider the following radial sorting procedure of $S$. 

\textsc{RadiallySort}($S, \mu, \delta, \wp$): 
Sort points radially around $\wp$, starting at $\mu$ in the direction of $\delta$ with respect to the line through $\wp$ and $\mu$ oriented from $\wp$ to $\mu$. In other words, if $\delta$ is to the left of $\ray{\wp\mu}$, then turn counter-clockwise. Else, turn clockwise. 
 See Figure~\ref{fig:ord-radial}(left), where we turn counter-clockwise and in~\ref{fig:ord-radial}(right) we turn clockwise. 

Using Observations~\ref{obs:direct} and~\ref{obs:closest} we can adapt  \textsc{RadiallySort} to be affine invariant and obtain a radial order of $S$. In other words, there exist three affine invariant function $f_1, f_2, f_3$ such that for any invertible affine transformation $\alpha$, the following holds: consider \textsc{RadiallySort}($S, f_1(S), f_2(S) f_3(S)$)$=p_0, \ldots p_{n-1}$ and  \textsc{RadiallySort}($\alpha(S), f_1(\alpha(S)), f_2(\alpha(S)),$ $f_3(\alpha(S))$)$=p'_0, \ldots p'_{n-1}$, then $\alpha(p_i)=p'_i$ for all $i \in \{0, \ldots, n-1\}$.

\begin{figure}
\centering
	\includegraphics[scale=1]{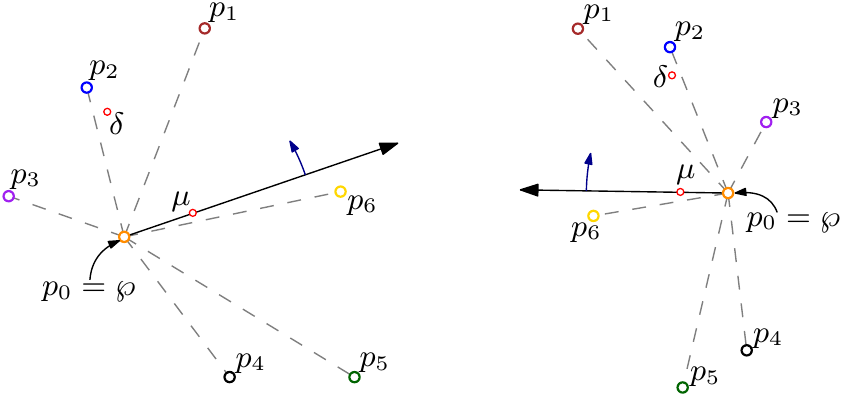}
	\caption{An affine invariant ordering of the point sets $S$ (\emph{left}) and its affine transformation $S'$$=$$\alpha(S)$ (\emph{right}). }
	\label{fig:ord-radial}
\end{figure}

\begin{theorem}\label{thm:graham}
Let $S$ be a point set in general position and let $f_1(S)=\mu,f_2(S)=\delta$ and $f_3(S)=\wp$ be three affine invariant functions such that $\mu, \delta, \wp$ are three non-collinear points and $\wp \in S$. Then, \textsc{RadiallySort($S, \mu, \delta, \wp$)} is an affine invariant radial ordering method of $S$. 
\end{theorem}

\begin{proof}
Let $\alpha$ be an affine transformation. 
 Let $p_0=\wp, p_2, \ldots, p_{n-1}$ be the order of $S$ given by
  \textsc{RadiallySort($S, \mu, \delta, \wp$)} and let $p'_0=\wp',p'_2, \ldots, p'_{n-1}$ be the order of $\alpha(S)$ given by \textsc{RadiallySort($\alpha(S), f_1(\alpha(S))=\mu', f_2(\alpha(S))=\delta', f_3(\alpha(S))=\wp'$)}, respectively. Then, $\alpha(p_0)=\alpha(\wp)=\wp'=p_0, \alpha(\mu)=\mu'$ and $\alpha(\delta)=\delta'$. 
It remains to show that $p'_i = \alpha(p_i)$ for all $1 \leq i \leq n-1$.
If $p_1$ is on the ray $\overrightarrow{\wp \mu}$, then $\alpha(p_1)$ is on the ray $\overrightarrow{\wp'\mu'}$. Thus, $p'_1=\alpha(p_1)$. Assume that $p_1$ is not on $\overrightarrow{\wp\mu}$. 

 \begin{figure}[tb]
\centering
	\includegraphics[scale=1]{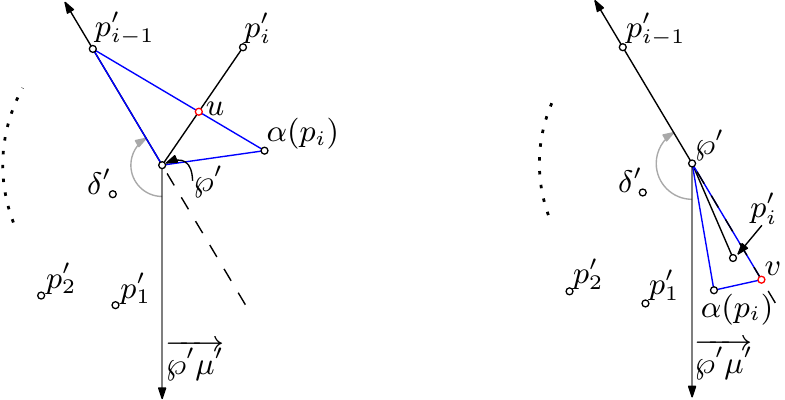}
	\caption{(\emph{left}) $p'_i$ is not in $T=\triangle(p'_{i-1}\wp'\alpha(p_i))$ and $u$ is the intersecting point of $\wp'p'_i$ and $\alpha(p_i)p'_{i-1}$. (\emph{right}) $p'_i$ is in the interior of $T=\triangle(v\wp'\alpha(p_i))$ and the Euclidean distance from $\wp'$ to $v$ is the same from $\wp'$ to $p'_{i-1}$.} 
	\label{fig:rad-proof}
\end{figure}

Let $i \in \{1, \ldots n-1\}$ and assume that $p'_j=\alpha(p_j)$ for all $j \leq i-1$. Let us show $\alpha(p_i)=p'_i$. Assume for the sake of a contradiction that $\alpha(p_i)\neq p'_i$. First, we show that $p'_i$ is on the same side as $\alpha(p_i)$ with respect to $\ray{\wp'p_{i-1}'}$ (we consider the ray $\ray{\wp'\mu'}$ when $i=1$).
Without loss of generality assume that $\delta'$ is to the right of $\overrightarrow{\wp' \mu'}$, since the case when $\delta'$ is to the left of $\overrightarrow{\wp' \mu'}$ is symmetric. 
If $\alpha(p_i)$ is to the right of $\ray{\wp'p_{i-1}'}$, then $p'_i$ has to be to the right of $\ray{\wp'p_{i-1}'}$ since $p'_i$ is hit first while rotating $\ray{\wp'p_{i-1}'}$ counter-clockwise. If $\alpha(p_i)$ is to the left of $\ray{\wp'p_{i-1}'}$, then by Observation~\ref{obs:direct} we have that $p_i$ is to the left (resp., \emph{to the right}) of  $\ray{\wp p_{i-1}}$ if $\delta$ is to the right (resp., \emph{to the left}) of $\ray{\wp\mu}$. Thus, there is no point of $S$ to the right (resp., to the left) of $\ray{\wp p_{i-1}}$ if $\delta$ is to the left (resp., to the right) of $\ray{\wp\mu}$. Hence, by Observation~\ref{obs:direct}, there is no point to the left of $\ray{\wp'p_{i-1}'}$, in particular, $p'_i$ cannot be on that side.
Therefore, $\alpha(p_i)$ and $p'_i$ lie on the same side with respect to $\overrightarrow{\wp'p_{i-1}'}$.

Let $\ell$ be the line containing $\wp'$ and $p'_{i-1}$ and let $v$ be the point on $\ell$ on opposite direction of $\overrightarrow{\wp'p'_{i-1}}$ with respect to $\wp'$, such that the length from $\wp'$ to $v$ is the same as the length from $\wp'$ to $p'_{i-1}$. See Figure~\ref{fig:rad-proof}(right).

We define a triangle $T$ in the following fashion.
\begin{compactitem}
\item If  $\alpha(p_i)$ and $p'_i$ are to the right of $\overrightarrow{\wp'p'_{i-1}}$, then $T=\triangle(p'_{i-1}\wp'\alpha(p_i))$. 
\item If $\alpha(p_i)$ and $p'_i$ are to the left of $\overrightarrow{\wp'p'_{i-1}}$, then $T=\triangle(v\wp'\alpha(p_i))$.
\end{compactitem}

Hence, either $p'_i$ is in $T$ or not. See Figure~\ref{fig:rad-proof}. Again, 
if $p'_i$ is contained in $T$, then from Proposition~\ref{prop:properties}(3) it follows that $\alpha^{-1}(p'_i)$ is contained in $T$, which contradicts that $p_i$ appears immediately after $p_{i-1}$ when sorting $S$ around $\wp$. On the other hand, if $p'_i$ is not in $T$, then the line segment $\wp'p'_i$ intersects $T$ on the edge opposite to $\wp'$. Let $u$ be the intersection point of $\wp'p'_i$ and the edge of $T$ opposite to $\wp'$. The point $u$ is in triangle $T$ and in the line segment $\wp'p'_1$. Thus, by Properties 1 and 3 of Proposition~\ref{prop:properties}, we have that $\alpha^{-1}(u)$ is in $\alpha(T)$ and on the line segment $\wp\alpha^{-1}(p'_i)$, which again contradicts the fact that $p_i$ appears first while ordering $S$ around $\wp$. Therefore, $p'_i=\alpha(p_i)$.
\end{proof}

\subsection{Affine invariant sweep-line ordering}\label{sec:sweep-line}

Let $\mu, \delta$ and $\wp$ be three non-collinear points, such that $\wp$ is in $S$. Consider the following sweep-line order method.

\textsc{SweepLine($S, \mu, \delta, \wp$):} 
Let $\ell^{\perp}$ be the line containing $\mu$ and $\wp$. Let $t$ and $b$ be a farthest point from the line $\ell^{\perp}$ that lies on the same and opposite side of $\delta$ with respect to  $\ell^{\perp}$, respectively.  
Then, sort the vertices in $S$ by sweeping with lines parallel to $\ell^{\perp}$ from $t$ towards $b$. If two vertices $u,v$ lie on the same line parallel to $\ell^{\perp}$, we say that $v$ appears before $u$ if and only if $\delta$ is to the left of $\overrightarrow{\wp\mu}$ and $v$ is to the left of $\overrightarrow{u\mu}$. 

Note that a consequence of Observation~\ref{obs:direct} is that if two vertices of $P$ lie on the same side with respect to $\ell^{\perp}$, then the corresponding vertices in $\alpha(P)$ lie again on the same side with respect to $\alpha(\ell^{\perp})$.

\begin{observation}\label{obs:side}
Let $\ell$ be a line in $\R^2$ and let $v, u$ be two points in $S$. Let $\alpha$ be an affine transformation. If $u$ and $v$ are on the same side of $\ell$, then $\alpha(u)$ and $\alpha(v)$ are on the same side of $\alpha(\ell)$. 
\end{observation}

The following theorem shows that such an ordering method is affine invariant when the primitives are affine invariant. See Figure~\ref{fig:sweep-line}. 

  \begin{figure}
\centering
	\includegraphics[scale=1]{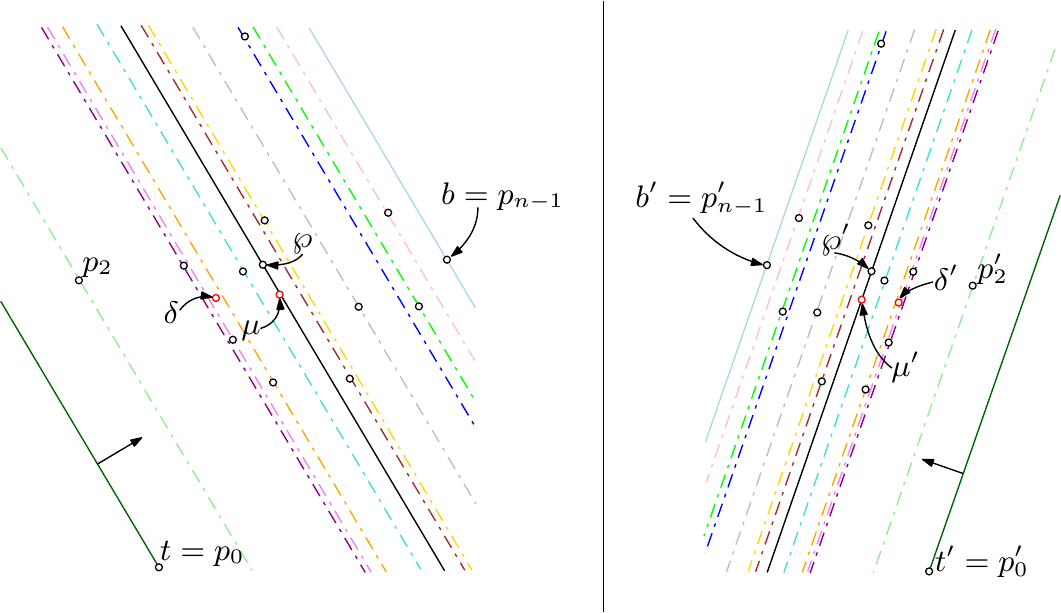}
	\caption{Each colored line is a line parallel to the black line containing the corresponding $f_3(S)=\wp$ and $f_1(S)=\mu$. Each color corresponds to the same line in the transformation $\alpha$. \text{SweepLine($S, \mu, \delta, \wp$)} is affine invariant.} 
	\label{fig:sweep-line}
\end{figure}

\begin{theorem}\label{thm:sweeping}
Let $S$ be a point set in general position and let $f_1(S)=\mu,f_2(S)=\delta$ and $f_3(S)=\wp$ be three affine invariant functions such that $\mu, \delta, \wp$ are three non-collinear points and $\wp \in S$. Then, \textsc{SweepLine($S, \mu, \delta, \wp$)} is an affine invariant sweep-line ordering method of $S$.
\end{theorem}
\begin{proof}
Let $\alpha(x)=Mx+b$ be an affine transformation.  
Let $f_1(\alpha(S))=\mu'$, $f_2(\alpha(S))=\delta', f_3(\alpha(S))=\wp'$. Then,  
$\alpha(\mu)=\mu', \alpha(\wp)=\wp', \alpha(\delta)=\delta'$. Let $p_0, p_1, \ldots, p_{n-1}$ be the resulting order of \textsc{SweepLine}($S, \mu, \delta, \wp$) and let  $p'_0,$ $p'_1, \ldots, p'_{n-1}$ be the resulting order of \textsc{SweepLine}($\alpha(S), \mu', \delta', \wp'$). Let us show that $\alpha(p_i)=p'_i$ for all $i \in \{0, 1, \ldots, n-1\}$.

Let $\ell'^{\perp}$ be the line containing $\mu'$ and $\wp'$. From Proposition~\ref{prop:properties}(1), we have that $\alpha(\ell^{\perp})=\ell'^{\perp}$. We denote by $\ell'^{\perp}_v$ the line parallel to $\ell'^{\perp}$ containing the point $v \in S$. Let $t' \in \alpha(S)$ be a farthest point from $\ell'^{\perp}$ that lies on the same side of $\delta'$ with respect to $\ell'^{\perp}$ and let $b'$ be a farthest point from $\ell'^{\perp}$ that lies on the opposite side of $\delta'$ with respect to $\ell'^{\perp}$. Thus, all points of $\alpha(S)$ that are not on  $\ell'^{\perp}_{t'}$ lie on the same side of $\ell'^{\perp}_{t'}$. Also, all points of $\alpha(S)$ that are not on  $\ell'^{\perp}_{b'}$ lie on the same side of $\ell'^{\perp}_{b'}$. 

On the other hand, $t$ and $\delta$ are on the same side of $\ell^{\perp}$. Thus, by Observation~\ref{obs:side}, it follows that $\alpha(t)$ and $\delta'$ are on the same side of $\ell'^{\perp}$. Since all points of $S$ that are not on  $\ell^{\perp}_{t}$ lie on the same side of $\ell^{\perp}_{t}$, by Observation~\ref{obs:side} it follows that all points of $\alpha(S)$ that are not on  $\ell'^{\perp}_{\alpha(t)}$ lie on the same side of $\ell'^{\perp}_{\alpha(t)}$. Since both $\alpha(t)$ and $t'$ lie on the same side of $\delta'$ with respect to $\ell'^{\perp}$, it follows that both lines $\ell'^{\perp}_{t'}$ and $\ell'^{\perp}_{\alpha(t)}$ lie on the same side with respect to $\ell'^{\perp}$. Finally, since there is exactly one line parallel to $\ell'^{\perp}$ on the same side of $\delta'$ with respect to $\ell'^{\perp}$ that passes through a point of $S$ such that all points of $S$ that are not in such line are on the same side, it follows that
 $\ell'^{\perp}_{\alpha(t)}=\ell'^{\perp}_{t'}$. 
 
 Similarly, since all points of $S$ that are not on  $\ell^{\perp}_{b}$ lie on the same side of $\ell^{\perp}_{b}$, by Observation~\ref{obs:side} it follows that all points of $\alpha(S)$ that are not on  $\ell'^{\perp}_{\alpha(b)}$ lie on the same side of $\ell'^{\perp}_{\alpha(b)}$. In addition, since $b$ is not on $\ell^{\perp}_{t}$, $\alpha(b)$ is not on $\ell'^{\perp}_{t'}$. Therefore, $\ell'^{\perp}_{\alpha(b)}=\ell'^{\perp}_{b'}$. Thus, the initial and the final lines while sweeping $S$ with lines parallel to $\ell^{\perp}$ are mapped to the initial and final lines while sweeping $S'$ with lines parallel to $\ell'^{\perp}$. Assume that the order of the lines parallel to $\ell^{\perp}$ up to the $(i-1)$-th line while sweeping $S$ is preserved under affine transformations. Let $\ell^{\perp}_i$ be the $i$-th line parallel to $\ell^{\perp}$ while sweeping $S$. From Observation~\ref{obs:side} it follows that for each point $v$ of $S$ that lies on the same side of $\ell^{\perp}_i$ as $t$, $\alpha(v)$  lies on the same side of $\alpha(\ell^{\perp}_i)$ as $t'$. Since all points of $S$ that are on the same side  of $\ell^{\perp}_i$ as $t$ have been swept by the first $i-1$ lines parallel to $\ell^{\perp}$, it follows that $\alpha(\ell^{\perp}_i)$ is the $i$-th line parallel to $\ell'^{\perp}$ while sweeping $\alpha(S)$. 

It remains to prove that \textsc{SweepLine($S, \mu, \delta, \wp$)} is affine invariant when two points $v$ and $u$ of $S$ lie on the same line parallel to $\ell^{\perp}$. Let $v$ and $u$ be two points in $S$ that lie on the same line parallel to $\ell^{\perp}$ such that $v$ appears before $u$ while ordering $S$. From Proposition~\ref{prop:properties}(1), $\alpha(u)$ and $\alpha(v)$ lie on the same line parallel to $\ell'^{\perp}$. Assume $\delta$ is to the left of $\overrightarrow{\wp\mu}$, then $v$ is to the left of $\overrightarrow{u\mu}$. Consider the following cases.
 
 \emph{Case 1)} If ${det}(M)>0$, then $\delta'$ is to the left of $\overrightarrow{\wp'\mu'}$ and $\alpha(v)$ is to the left of $\overrightarrow{\alpha(u)\mu'}$. Therefore, $\alpha(v)$ appears before $\alpha(u)$. 
 
 \emph{Case 2)} If ${det}(M)<0$, then $\delta'$ is to the right of $\overrightarrow{\wp'\mu'}$ and $\alpha(v)$ is to the right of $\overrightarrow{\alpha(u)\mu'}$. Therefore, $\alpha(v)$ appears before $\alpha(u)$.
 
 The case when $\delta$ is to the right of $\overrightarrow{\wp\mu}$ is symmetric.
\end{proof}

\section{Applications to affine invariant geometric objects}\label{sec:applications}

In this section we give different applications of the $A_S$-norm and methods given in Section~\ref{sec:sorting} that result in affine invariant algorithms. 

\subsection{Affine invariant algorithms based on an $A_S$-disk}\label{subsec:disk}

In this section we mention different geometric objects of a point set $S$ that satisfy an empty $N$-disk property: for each pair $p_i, p_j \in S$, edge $p_ip_j$ is added if there exists a homothet of the $N$-disk that contains $p_i$ and $p_j$ on its boundary and no other point of $S$ in its interior. Using the fact that the $A_S$-norm is affine invariant, then the empty $A_S$-disk property is also affine invariant. Thus, using the $A_S$-disk gives us a set of affine invariant algorithms for computing more geometric objects.  

Let $N$ be a normed metric and let $S$ be a point set in general position. 
The \emph{$N$-Gabriel graph} of $S$, denoted $GG_{N}[S]$, is defined with $S$ as the vertex set and for each pair of points $u, v \in S$, the edge $uv$ is in $GG_N[S]$ if and only if there exists an $N$-disk of radius $N(u$$-$$v)/2$ containing $u$ and $v$ on its boundary and no other point of $S$. In other words, the edge $uv$ is in $GG_N[S]$ if a smallest $N$-disk containing $u$ and $v$ contains no other point of $S$. If the boundary of the $N$-disk is defined by a smooth convex shape, then the $N$-disk of radius $N(u$$-$$v)/2$ containing $u$ and $v$ is unique.
The \emph{$N$-relative neighborhood graph} of $S$, denoted $RNG_N[S]$, is the graph with vertex set $S$ where for every pair of points $u,v \in S$, the edge $uv$ is in $RNG_N[S]$ if and only if $N(u$$-$$v)\leq\max\{N(u$$-$$w), N(v$$-$$w)\}$ for any $w \in S$. Geometrically, for every pair of points $u,v \in S$, the edge $uv$ is in $RNG_N[S]$ if and only if the intersection of the $N$-disks with radius $N(u-v)$ centered at $u$ and $v$, respectively, does not contain points of $S$ in its interior. 
 The \emph{order-$k$ $N$-Delaunay graph} of $S$, denoted $k$-${DG}_{N}[S]$, is the graph with vertex set $S$ and has an edge $uv$ provided that there exists an  $N$-disk with $u$ and $v$ on its boundary enclosing at most $k$ points of $S$ different from $u$ and $v$. Similarly, one can define the \emph{order-$k$ $N$-Gabriel graph} and the \emph{order-$k$ $N$-relative neighborhood graph} of $S$. 
 
 Since the $A_S$-disks are affine invariant, we obtain the following corollary. 
 
 \begin{corollary}
 Let $S$ be a point set in general position. Then, the following geometric structures are affine invariant: $GG_{A_S}[S]$, ${RNG}_{A_S}[S]$, $k$-${DG}_{A_S}[S]$, $k$-${GG}_{A_S}[S]$, $k$-${RNG}_{A_S}[S]$.
 \end{corollary}
 
 There is a hierarchical relation between these structures the we will give in more detail in Subsection~\ref{subsec:norm}. 

We conclude this subsection with the following lemma. 

\begin{lemma}
Any algorithm that computes a geometric object based on an $A_S$-disk is affine invariant.
\end{lemma}

\subsection{Affine invariant algorithms based on an $A_S$-norm}\label{subsec:norm}

Consider a norm $N$. In this section we mention different geometric objects of a point set $S$ constructed using the $N$-length of all possible $n\choose2$ edges in $S$. Using the fact that the $A_S$-norm is affine invariant, the order of the edges between each pair of points in $S$ given by the $A_S$-norm is also affine invariant. Therefore, using such ordering gives us a set of affine invariant algorithms for computing geometric objects. 

The first object is the \emph{$N$-minimum weight triangulation} of $S$, denoted ${MWT}_N[S]$, which minimizes the sum of the $N$-length of its edges. The \emph{$N$-Greedy triangulation} of $S$, denoted ${\emph{GreedyT}}_{N}[S]$ is a triangulation that adds at each step an edge in strict increasing order of the $N$-length of the edges provided that the new edge does not intersect in the interior of a previously inserted edge. 
The \emph{$N$-minimum spanning tree} of $S$, denoted $MST_N[S]$, is a spanning tree of $G$ with minimum total edge $N$-length. The \emph{$N$-closest pair} of $S$ computes the edge with smallest $N$-length. The \emph{$k$-$N$-nearest neighbor graph} of $S$, denoted $k$-$NNG_N[S]$, is the graph with vertex set $S$ and an edge $uv \in k$-$NNG_N[S]$ provided that the $N$-distance between $u$ and $v$ is the $k$-th smallest $N$-distance from $u$ to any other point in $S$. 

We show the following lemma using the fact that the $A_S$-norm is affine invariant.

\begin{lemma}\label{lemma:edge-length-order}
There exists an affine invariant order of the $n\choose2$ edges in any point set in general position.
\end{lemma}
\begin{proof}
Consider an affine transformation $\alpha$, since $A_{S}(u-v)=A_{\alpha(S)}(\alpha(u)-\alpha(v))$, we obtain that if $A_{S}(u-v)\leq A_{S}(u'-v')$, then $A_{\alpha(S)}(\alpha(u)-\alpha(v))\leq A_{\alpha(S)}(\alpha(u')-\alpha(v'))$. Hence, the increasing and decreasing order of the $A_S$-length of the edges in $S$ is affine invariant, i.e., the $i$-th $A_S$-largest (resp. $A_S$-shortest) edge in $S$ is mapped by $\alpha$ to the $i$-th $A_{\alpha(S)}$-largest (resp. $A_{\alpha(S)}$-shortest) edge of $\alpha(S)$.
\end{proof}

Lemma~\ref{lemma:edge-length-order} implies the following corollaries.

\begin{corollary}
Any greedy algorithm defined by the $A_{S}$-length of the edges of $S$ is affine invariant.
\end{corollary}

 \begin{corollary}
 Let $S$ be a point set in general position. The following objects are affine invariant: ${MWT}_{A_S}[S], {\emph{GreedyT}}_{A_S}[S], MST_S[S],$ and $k$-$NNG_N[S]$.
 \end{corollary}

Consider an $N$-norm such that boundary of the $N$-disk defines a convex shape in the plane. Aurenhammer and Paulini~\cite{aurenhammer2014} proved the following theorem.

\begin{theorem}[\cite{aurenhammer2014}]\label{thm:aurenhammer}
 For point set in the plane and any $N$-norm with convex $N$-disk: 
 
 $MST_N[S] \subseteq RNG_N[S] \subseteq GG_N[S] \subseteq DT_N[S]$. 
 \end{theorem}
 
Their proof follows from the definition of each structure and it can be extended for  $MST_N[S] \subseteq k$-$RNG_N[S] \subseteq k$-$GG_N[S] \subseteq k$-$DT_N[S]$. Thus, ${MST}_{A_S}[S]\subseteq k$-${RNG}_{A_S}[S]\subseteq k$-${GG}_{A_S}[S]\subseteq k$-${DT}_{A_S}[S]$.

\subsection{Affine invariant geometric objects using radial ordering}

In this Section we give two applications of Theorem~\ref{thm:graham} for computing affine invariant triangulations of a point set and a method that computes an affine invariant quadrangulation of a point set.

\subsubsection{An affine invariant Graham triangulation}\label{sec:graham}
One of the most popular algorithms for computing the convex hull of a point set $S$ is Graham's scan~\cite{graham1972efficient}. A nice property of this algorithm is that a modification of the algorithm can also produce a triangulation~\cite{toussaint1982convex}, sometimes called \emph{Graham's triangulation}~\cite{monroy2005graham}. In this section we present an affine invariant version of Graham's scan using the sorting method \textsc{RadiallySort}.

Our method is based on the algorithm for Graham triangulation~\cite{toussaint1982convex}. 
 Informally, the method consists in first choosing a point $p$ and then radially sorting the remaining points around $p$, say $p_1, \ldots, p_{n-1}$. Once the points are sorted, the algorithm adds edge $pp_i$ for all $i\in \{1, \ldots, n-1\}$. Finally, the algorithm visits each remaining point $p_i$ in order, and adds edge $p_ip_j$ if and only if  $p_j$ is visible~\footnote{A point $p$ is \emph{visible} from $p$ in a graph $G$ if one of the following holds: (a) the edge $pq$ is in $G$, or (b) the line segment $pq$ does not cross any edge of $G$ in its interior.}  
from $p_i$, for all $j>i$ (see~\cite{toussaint1982convex} for details and proof of correctness).

The Graham triangulation can be computed in $O(n)$ time when $S$ is given radially ordered. Moreover, since the edges of the triangulation are added according to the radial sort of $S$, it follows that if the radial order is affine invariant, then the triangulation is affine invariant.

The following result is an implication of Theorem~\ref{thm:graham}. 

\begin{corollary}
There exists an affine invariant Graham triangulation for any point set $S$ in general position that runs in $O(n)$ time after sorting.
\end{corollary}
\begin{proof}
Let $S$ be a point set in general position, and let $\alpha(S)=S'$ be an affine transformation of $S$. Let $f_1(S)=\mu,f_2(S)=\delta$ and $f_3(S)=\wp$ be three affine invariant functions such that $\mu, \delta, \wp$ are three non-collinear points and $\wp \in S$. Hence, $f_1(S')=\alpha(\mu), f_2(S')=\alpha(\delta)$, and $f_3(S')=\alpha(\wp)$. Let \textsc{RadiallySort($S, \mu, \delta, \wp$)}$=\{p_0=\wp, p_1, \ldots, p_{n-1}\}$ and \textsc{RadiallySort($S', \alpha(\mu),$ $\alpha(\delta), \alpha(\wp)$)} $=\{p'_0=\alpha(\wp), p'_1, \ldots, p'_{n-1}\}$. Let $T(S)$ and $T(S')$ be the resulting Graham triangulations of the radially sorted $S$ and $S'$, respectively. Note that it remains to show that for each edge $p_ip_j \in T(S)$ the edge $\alpha(p_i)\alpha(p_j) \in T(S')$: We have that  if $\triangle(p_ip_jp_k) \in T(S)$, then the edges of $\triangle(\alpha(p_i)\alpha(p_j)\alpha(p_k))$ are in  $T(S')$. Moreover, by Proposition~\ref{prop:properties}(\ref{polygons}), no point of $S$ is mapped in the interior of $\triangle(\alpha(p_i)\alpha(p_j)\alpha(p_k))$. Thus, $\triangle(\alpha(p_i)\alpha(p_j)\alpha(p_k)) \in T(S')$. Hence, let us show that for each edge $p_ip_j \in T(S)$, the edge $\alpha(p_i)\alpha(p_j) \in T(S')$. Assume $0\leq i<j\leq n-1$. We know by Theorem~\ref{thm:graham} that $p'_i=\alpha(p_i)$ and $p'_j=\alpha(p_j)$. So the algorithm visits each point in an affine invariant order. If $i=0$, then by construction $p'_0p'_j \in T(S')$. For the sake of a contradiction suppose $i>0$ is the first index such that $p'_ip'_j \notin T(S')$. Then, when visiting $p'_i$ in the algorithm, $p'_j$ is not visible to $p'_i$. Hence, there exists an edge $p'_{k}p'_t \in T(S')$ that intersects $p'_ip'_j$ in its interior, such that $k<i$. Thus, $p_kp_t \in T(S)$. By Proposition~\ref{prop:properties}(\ref{lines}) line segments are mapped to line segments, thus $p_ip_j$ intersects with $p_kp_t$, a contradiction. Therefore, $p'_ip'_j \in T(S')$.

The running time follows from the algorithms of Avis and Toussaint~\cite{toussaint1982convex}. 
\end{proof}

\subsubsection{Affine invariant Hamiltonian triangulations}\label{subsec:hamiltonian}

When the point set $S$ has at least one point in the interior of its convex hull, then $S$ can be triangulated by the \emph{insertion method}, which consists of computing the convex hull of $S$ and then inserting points from the interior in arbitrary order. Every time a point $v$ is inserted, the edges connecting
 $v$ with the points defining the face that contains $v$ are added. 
 
 Applying the insertion method to $S$ such that the interior points in the convex hull of $S$ are inserted in the order given by \textsc{RadiallySort} results in an affine invariant algorithm that computes a triangulation of $S$. 
 
 \begin{corollary}
 There exists an affine invariant insertion method that computes a triangulation for any point set in general position with at least one point in the interior of its convex hull.
 \end{corollary}
  \begin{proof}
  Let $S$ be a point set in general position, and let $\alpha(S)=S'$ be an affine transformation of $S$. By Proposition~\ref{prop:properties}(\ref{polygons}), the convex hull of $S$, denoted $CH(S)$, is mapped to the convex hull of $S'$, denoted $CH(S')$. We consider the point sets $S\setminus CH(S)$ and $S'\setminus CH(S')$, both of size $1\leq m\leq n-3$.  Let $f_1(S\setminus CH(S))=\mu,f_2(S\setminus CH(S))=\delta$ and $f_3(S\setminus CH(S))=\wp$ be three affine invariant functions such that $\mu, \delta, \wp$ are three non-collinear points and $\wp \in S\setminus CH(S)$. Let \textsc{RadiallySort($S\setminus CH(S), \mu, \delta, \wp$)}$=\{p_0=\wp, p_1, \ldots, p_{m-1}\}$ and \textsc{RadiallySort($S'\setminus CH(S'), f_1(S'\setminus CH(S'))=\alpha(\mu), f_2(S'\setminus CH(S'))=\alpha(\delta), f_3(S'\setminus CH(S'))=\alpha(\wp)$)} $=\{p'_0=\alpha(\wp), p'_1, \ldots, p'_{m-1}\}$. Let $T(S)$ and $T(S')$ be the resulting triangulations using the insertion method with the radially sorted $S\setminus CH(S)$ and $S'\setminus CH(S')$, respectively. From Theorem~\ref{thm:graham} we have that $p'_i=\alpha(p_i)$ for all $i \in \{0, \ldots, m-1\}$. Thus, we insert the points in an affine invariant order. Assume that we have inserted $i$ points with $i \in \{0, \ldots m-1\}$ and that the resulting triangulations $T_i$ and $T'_i$ of $S\setminus\{p_i, \ldots p_{m-1}\}$ and $S'\setminus\{p'_i, \ldots, p'_{m-1}\}$, respectively, can be mapped into each other by $\alpha$. Consider the polygon $P_i$ of $T_i$ that contains $p_i$ in its interior. From Proposition~\ref{prop:properties}(\ref{polygons}) it follows that $p'_i$ is in the interior of $\alpha(P_i) \in T'_i$. Thus,  when inserting $p_i$, each new edge $p_iv$ with $v \in P_i$, corresponds to the new edge $p'_i\alpha(v)$ when inserting $p'_i$ in $T'_i$. Consequently, $T(S)$ is isomorphic to $T(\alpha(S))$ with respect to $\alpha$. 
\end{proof}

   These triangulations (Graham triangulation and the triangulation by insertion method) are Hamiltonian, i.e., their duals~\footnote{The \emph{dual} of a graph $G$ has as vertices the faces of $G$ and an edge between two vertices if and only if the two faces share an edge in $G$.} 
contain a Hamiltonian path, and are of interest for fast rendering in computer graphics~\cite{arkin1996hamiltonian, bose1997characterizing, monroy2005graham}.

\begin{corollary}
There exist affine invariant methods that compute triangulations whose duals are Hamiltonian in $O(n\log n)$ time for any set of $n$ points in the plane in general position.
\end{corollary}

\subsubsection{An affine invariant quadrangulation}\label{subsec:quadrangulations}

A \emph{quadrangulation} of a point set $S$ in the plane is a geometric graph such that its vertices are the points of $S$, its edges are line segments joining vertices and all of its faces (except possibly the exterior face) are quadrilaterals. In this section we present an affine invariant algorithm that computes a quadrangulation of $S$ when $S$ has even number of points in its convex hull.

When $S$ is in general position and $S$ has an even number of points in its convex hull, Bose and Toussaint~\cite{bose1997characterizing} provide a $O(n)$ time algorithm after sorting that quadrangulates $S$ using an insertion method similar to the one discussed in Section~\ref{subsec:hamiltonian}, described as follows.  Let $CH(S)=\{p_0, \ldots p_{m-1}\}$, ordered in clockwise order, and let $S\setminus CH(S) = \{p_{m}, \ldots, p_{n-1}\}$ denote the points of $S$ in the interior of $CH(S)$. First, add edges $p_0p_k$ where $ 3 \leq k\leq m-2$ and $k\equiv 1 \mod 2$. This produces a quadrangulation of $CH(S)$. Thus, if $CH(S)=S$, the algorithm finishes. Otherwise, we insert the points in $S\setminus CH(S)$ in order in the following way. At iteration $i+1$, with $m-1<i\leq n-1$, insert $p_i$, and let $D=\square(p_jp_kp_tp_r)$ be the quadrangle that contains $p_i$ such that $0\leq j<k,t,r$. If $D$ is convex, then add the edges $p_ip_j$ and $p_ip_t$. If  $D$ is not convex, then add edges $p_iv$ and $p_iv'$ where $v$ is a the unique reflex vertex in $D$ and $v'$ is the vertex in $D$ that is not adjacent to $v$. 

Note that if the points of $S$ are given in an affine invariant order such that the first $|CH(S)|$ points correspond to the ones of $CH(S)$, then this is an affine invariant insertion method that produces quadrangulation. Using Theorem~\ref{thm:graham} we get the following corollary.

\begin{corollary}\label{cor:quadrangulation}
There exists and affine invariant algorithm that computes a quadrangulation in $O(n)$ time after sorting for any point set in general position and with even number of points in its convex hull.
\end{corollary}
\begin{proof}
  Let $S$ be a point set in general position with $m$ even points in its convex hull $CH(S)$, and let $\alpha(S)=S'$ be an affine transformation of $S$. By Proposition~\ref{prop:properties}(\ref{polygons}), the convex hull of $S$, denoted $CH(S)$, is mapped to the convex hull of $S'$, denoted $CH(S')$. Thus, $|CH(S')|=m$. 
  
Let $f_1(CH(S))=\mu^{ch}, f_2(CH(S))=\delta^{ch}$ and $f_3(CH(S))=\wp^{ch}$ be three affine invariant functions such that $\mu^{ch}, \delta^{ch}, \wp^{ch}$ are three non-collinear points and $\wp^{ch} \in CH(S)$. Let $f_1(S\setminus CH(S))=\mu,f_2(S\setminus CH(S))=\delta$ and $f_3(S\setminus CH(S))=\wp$ be three affine invariant functions such that $\mu, \delta, \wp$ are three non-collinear points and $\wp \in S\setminus CH(S)$. We sort $S:= \{p_{0}, \ldots p_{n-1}\}$ such that \textsc{RadiallySort($CH(S), \mu^{ch}, \delta^{ch}, \wp^{ch}$)}$=\{p_0=\wp^{ch}, p_1, \ldots, p_{m-1}\}$ and \textsc{RadiallySort($S\setminus CH(S), \mu, \delta, \wp$)}$=\{p_{m}=\wp, p_{m+1}, \ldots, p_{n-1}\}$. Consider the sorted 
  $S':=\{p'_0, \ldots p'_{n-1}\}$ such that 
\textsc{RadiallySort($CH(S'), f_1(CH(S'))=\alpha(\mu^{ch}), f_2(CH(S'))=\alpha(\delta^{ch}),$ $f_3(CH(S'))=\alpha(\wp^{ch})$)} $=\{p'_0=\alpha(\wp^{ch}), p'_1, \ldots, p'_{m-1}\}$ and \textsc{RadiallySort($S'\setminus CH(S'),$ $f_1(S'\setminus CH(S'))=\alpha(\mu), f_2(S'\setminus CH(S'))=\alpha(\delta), f_3(S'\setminus CH(S'))= \alpha(\wp)$)} $=\{p'_m=\alpha(\wp), p'_{m+1},$ $\ldots, p'_{n-1}\}$. From Theorem~\ref{thm:graham} it follows that $p'_i=\alpha(p_i)$ for all $i \in \{0, \ldots, n\}$. Since the points of $CH(S)$ and $CH(S')$ are ordered radially around the initial point $p_0$ and $p'_0$, respectively, it follows that for all $i\in \{0, \ldots m-1\}$, $p_ip_{i+1}$ and $p'_ip'_{i+1}$ is an edge in $CH(S)$ and $CH(S')$, respectively. Thus, the quadrangulations $Q(CH(S))$ and $Q(CH(S'))$ of the sorted sets $CH(S)$ and $CH(S')$ are isomorphic with respect to $\alpha$, since $p'_0p'_k=\alpha(p_0)\alpha(p_k)$ for all $ 3 \leq k\leq m-2$ and $k\equiv1\mod 2$. 

Note that if $S$ does not have points in the interior of $CH(S)$, then the statement follows. Hence, assume that $S$ contains at least one point in the interior of its convex hull. Consider the quadrangulations $Q(S)$ and $Q(S')$ of the sorted $S$ and $S'$, respectively, using the insertion method. It remains to show that $\alpha(Q(S))=Q(\alpha(S))$. Assume that the first $m<i$ points of $S$ have been inserted and that the quadrangulations $Q_i=Q(S\setminus\{p_i, \ldots, p_{n-1}\})$ and $Q'_i= Q(S'\setminus\{p'_i, \ldots, p'_{n-1}\})$ can be transformed into each other by $\alpha$. Let $D=\square(p_jp_kp_tp_r) \in Q_i$ be the quadrangle that contains $p_i$ where $j < k,t,r$. By Proposition~\ref{prop:properties}(\ref{polygons}) it follows that $p'_i=\alpha(p_i)$ is in the interior of $D'=\square(p'_jp'_kp'_tp'_r)$. Moreover, if $D$ is convex, then $D'$ is also convex. Thus, if $D$ is convex, then the new edges $p_ip_j$ and $p_ip_t$ added to $Q_i$ correspond to the edges $p'_ip'_j$ and $p'_ip'_t$ added to $Q'_i$. Finally, consider the case when $D$ is not convex. Then, there is exactly one reflex vertex in $D$, say $p_t$. Thus, $p_t \in \triangle(p_jp_kp_r)$. Proposition~\ref{prop:properties}(\ref{polygons}) implies that $\alpha(p_t)=p'_t \in \triangle(p'_jp'_kp'_r)$.  Hence, $p'_t$ is the only reflex vertex of $D'$. Therefore, the new edges $p_ip_t$ and $p_ip_j$ added to $Q_i$ correspond to the new edges $p'_ip'_t$ and $p'_ip'_j$ added to $Q'_i$. 
Consequently, $Q(S)$ is isomorphic to $Q(\alpha(S))$ with respect to $\alpha$. 

The runtime follows from Bose and Toussaint~\cite{bose1997characterizing}.
\end{proof}

%
\subsection{An affine invariant triangulation of a polygon by ear clipping}\label{sec:traversing}

In this section we provide an affine invariant method for triangulating any simple polygon with its vertex set in general position. In particular, in order to provide such method, we first define a traversal of a polygon that is affine invariant. 

An \emph{ear} of a polygon is a triangle formed by three consecutive vertices $p_1,p_2$ and $p_3$  such that the line segment $p_1p_3$ is a \emph{diagonal}\footnote{A \emph{diagonal} of a polygon is a line segment between two non-consecutive vertices that is totally contained inside the polygon.} of the polygon.  
It is a well-known fact that every simple polygon contains two ears (see Meisters~\cite{meisters1975polygons}). 
By recursively locating and chopping an ear, one can triangulate any simple polygon. This gives an $O(n^2)$ algorithm known as \emph{ear clipping} based on ElGindy et al.~\cite{elgindy1989slicing}. 

It follows from Properties 1 and 3 of Proposition~\ref{prop:properties} that the diagonals of a simple polygon $P$ are preserved under affine transformations. Thus, the ears of a  simple polygon are also preserved. Hence, if at every step the ear clipping procedure locates an ear by traversing the polygon in an affine invariant order, then such a procedure is affine invariant and computes a triangulation of $P$. The traversal of the polygon in an affine invariant order depends only on finding an affine invariant starting point, and on deciding correctly whether to traverse it clockwise or counter-clockwise. 
The latter depends only on whether the affine transformation contains an odd number of reflections. 

Let $P$ be a simple polygon with vertex sequence $S:=\{v_1, \ldots, v_n\}$, such that $S$ is in general position. Let $\mu, \delta$ and $\wp$ be three non-collinear points such that $\wp$ is in $S$. Consider the following traversal of $P$. 

\textsc{Traversal}($P, \mu, \delta, \wp$):  
If $\delta$ is to the left of $\overrightarrow{\mu \wp}$, then order $S$ by traversing $P$ from $\wp$ in counter-clockwise order.
Otherwise, order $S$ by traversing $P$ from $\wp$ in clockwise order. See Figure~\ref{fig:ordering-poly}.

Using arguments similar to the ones used to prove Theorem~\ref{thm:graham}, we show that \textsc{Traversal} is affine invariant.

\begin{theorem}\label{thm:trav}
Let $P$ be a simple polygon with vertex sequence $S:=\{v_1, \ldots,$ $v_n\}$, such that $S$ is in general position. Let $f_1(S)=\mu,f_2(S)=\delta$ and $f_3(S)=\wp$ be three affine invariant functions such that $\mu, \delta, \wp$ are three non-collinear points and $\wp \in S$. Then,
\textsc{Traversal($P,\mu,\delta,\wp$)} is an affine invariant traversal.
\end{theorem}
\begin{proof}

Let  $\alpha(P)$ be an affine transformation of $P$. The vertices of $\alpha(P)$ are $\alpha(S)$.  
Consider \textsc{Traversal($\alpha(P),f_1(\alpha(S))=\mu',f_2(\alpha(S))=\delta',f_3(\alpha(S))=\wp'$)}. 
By definition, we have that $\alpha(\mu)=\mu', \alpha(\wp)=\wp', \alpha(\delta)=\delta'$.  

Let $p_0=\wp, p_1, \ldots, p_{n-1}$ be the resulting order of \textsc{Traversal}($P,\mu,\delta,\wp$) and let $p'_0=\wp', p'_1, \ldots, p'_{n-1}$ be the resulting order of \textsc{Traversal}($\alpha(P),\mu',\delta',\wp'$). Let us show that $\alpha(p_i)=p'_i$ for all $i \in \{0, 1, \ldots, n-1\}$.

Since the ordering is defined while traversing the simple polygon, it suffices to show that $\alpha(p_0)=p'_0$ and that the vertices are traversed in an affine invariant order.  By definition, $p'_0=\wp'=\alpha(\wp)=\alpha(p_0)$. In addition, the orientation of the points of $P$ changes in $\alpha(P)$ when there is an odd number of reflections in $\alpha$. Hence, if the orientation of the points of $P$ changes, then $\delta'$ is on the opposite direction of $\ray{\wp'\mu'}$ as $\delta$ from $\ray{\wp\mu}$. Therefore, if the orientation of the points in $P$ is changed in $\alpha(P)$, then \textsc{Traversal}($\alpha(P),\mu',\delta',\wp'$) traverses $P'$ in the opposite direction \textsc{Traversal}($P,\mu,\delta,\wp$) traverses $P$. Thus, the vertices are traversed in an affine invariant order if the orientation of the points in $P$ is opposite to the one in $\alpha(P)$. Otherwise, if the orientation does not change, then  \textsc{Traversal}($\alpha(P),\mu',\delta',\wp'$) traverses the points clockwise if  and only if \textsc{Traversal}($P,\mu,\delta,\wp$) traverses the points clockwise.
\end{proof}
The following result is an implication of Theorem~\ref{thm:trav}.
\begin{corollary}
There exists an affine invariant ear clipping algorithm that computes a triangulation in $O(n^2)$ time for any simple polygon $P$ with vertex set in general position.
\end{corollary}
\begin{figure}
\centering
	\includegraphics[scale=1]{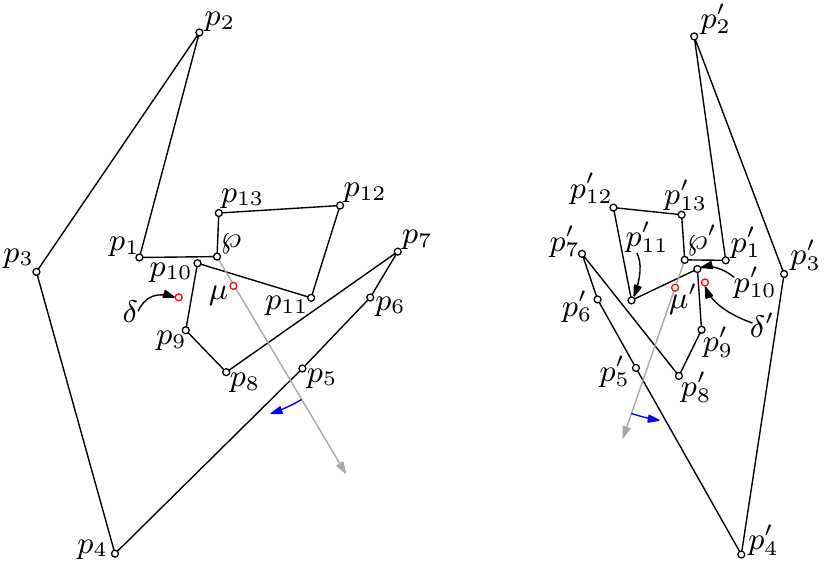}
	\caption{Each polygon can be transformed into each other by an affine transformation $\alpha$. Using affine invariant functions $f_1, f_2, f_3$ we obtain an affine invariant traversal, where $f_1(S)=\mu, f_2(S)=\delta, f_3(S)=\wp$ and $f_1(\alpha(S))=\mu', f_2(\alpha(S))=\delta', f_3(\alpha(S))=\wp'$.}
	\label{fig:ordering-poly}
\end{figure}


\subsection{An affine invariant triangulation of a polygon by sweep-line}\label{sucsec:sweeping}

In this section we give another affine invariant triangulation algorithm for any simple polygon. 

A simple polygon $P$ is \emph{monotone with respect to a line $\ell$} if for any line $\ell^{\perp}$ perpendicular to $\ell$, the intersection of $P$ with $\ell^{\perp}$ is connected. A line $\ell$ divides the plane into two half planes, and we say that two points \emph{lie on the same side of $\ell$} if they lie on the same half plane. Let $v$ be a vertex of $P$ and $\ell_{v}$ be the line containing $v$ that is parallel to $\ell$. 
We say that $v$ is an \emph{$\ell$-cusp} if $v$ is a reflex vertex of $P$ and its neighbors in $P$ lie on the same side of $\ell_v$. A characterization of a monotone polygon is stated in the following property.
\begin{property}[Garey et al.~\cite{garey1978triangulating}]
A polygon is $\ell$-monotone if and only if it does not contain an $\ell^{\perp}$-cusp.
\end{property}
An $\ell$-monotone polygon can be triangulated in $\Theta(n)$~\cite{garey1978triangulating} using Graham scan algorithm~\cite{toussaint1982convex} or by sweeping its vertices with line $\ell^{\perp}$~\cite{garey1978triangulating}.  In this section we use the latter. 

If $P$ is not $\ell$-monotone, then one can split $P$ into $\ell$-monotone subpolygons by adding diagonals in order to break all $\ell^{\perp}$-cusps. 
Then, each resulting $\ell$-monotone polygon can be triangulated independently. The way to add these diagonals is due to Lee et. al.~\cite{lee1977location} and also uses a sweep-line parallel to $\ell^{\perp}$.
For each vertex $v$ in $P$ take the line $\ell^{\perp}_v$ and partition $P$ into trapezoids defined by the intersection of the $\ell^{\perp}_v$ lines and the edges of $P$, this structure is called \emph{trapezoid decomposition} of $P$ and it can be computed in $O(n \log n)$ time using Garey et al.'s algorithm~\cite{garey1978triangulating} or in $O(n)$ time using Chazelle's algorithm~\cite{chazelle1991triangulating}.  Consider the vertex set $S$ of $P$ and two affine invariant functions $f_1, f_2$ such that $f_1(S)$ and $f_2(S)$ are points. Note that if we define $\ell^{\perp}$ as the line that passes through $f_1(S)$ and $f_2(S)$, then the trapezoid decomposition is affine invariant for any invertible affine transformation $\alpha$. Hence, we can adapt Garey et al.'s algorithm and Chazelle's algorithm in an affine invariant algorithm.

\begin{corollary}
There exists an affine invariant algorithm that computes a trapezoid decomposition in $O(n)$ time for any polygon $P$ with vertex set in general position. 
\end{corollary}

Once the trapezoid decomposition is constructed the polygon $P$ is partitioned in the following way. If a vertex $v$ is an $\ell^{\perp}$-cusp, add the diagonal of $P$ to the vertex $u$ that is not in the same side as the neighbors of $v$ with respect to $\ell^{\perp}_v$ and $u$ is the opposite vertex in a trapezoid containing $v$. 
See Figure~\ref{fig:monotone}. Hence, using these diagonals and then triangulating each $\ell$-monotone subpolygon of $P$, we obtain a triangulation of $P$. This is an $O(n \log n)$ algorithm based on Garey et al.~\cite{garey1978triangulating} and $O(n)$ algorithm based on Chazelle~\cite{chazelle1991triangulating}. 
  \begin{figure}
\centering
	\includegraphics[scale=1]{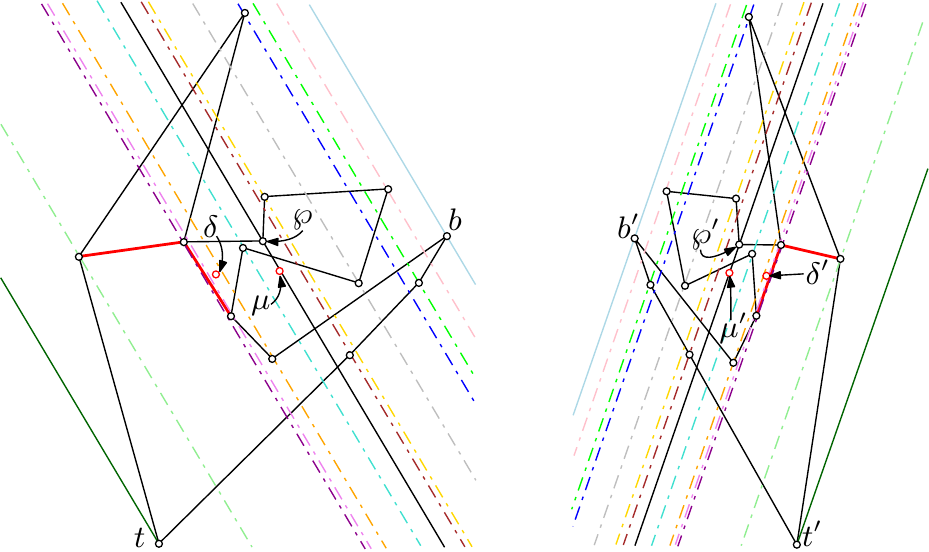}
	\caption{Each colored line is a line parallel to the black line. The red diagonals partition the simple polygon into monotone polygons. Each diagonal contains a cusp with respect to the black line. Each color corresponds to the same line in the transformation.} 
	\label{fig:monotone}
\end{figure}

Using Observation~\ref{obs:side}, we show that $\ell^{\perp}$-cusps are preserved under affine transformations.

\begin{lemma}\label{lem:cusps}
Let $\alpha(P)$ be an affine transformation of $P$. Each $\ell^{\perp}$-cusp of $P$ is mapped to an $\alpha(\ell^{\perp})$-cusp in $\alpha(P)$.
\end{lemma}
\begin{proof}
Let $v$ be an $\ell^{\perp}$-cusp of $P$. Let us show that $\alpha(v)$ is an $\alpha(\ell^{\perp})$-cusp in $\alpha(P)$. Let $v^{-}$ and $v^+$ be the neighbors of $v$ in $P$. If $v^-$ and $v^+$ are on the same side of $\ell^{\perp}_v$ then by Observation~\ref{obs:side} $\alpha(v^+)$ and $\alpha(v^-)$ lie on the same side of $\ell^{\perp}_v$. Moreover, the triangle $\triangle(v^-vv^+)$ is not in $P$. From Proposition~\ref{prop:properties}(3) follows that the triangle $\triangle(\alpha(v^-)\alpha(v)\alpha(v^+))$ is not in $\alpha(P)$. Hence, the interior angle of $\alpha(v)$ in $\alpha(P)$ is greater than $\pi$. So, $\alpha(v)$ is an $\alpha(\ell^{\perp})$-cusp in $\alpha(P)$. 
\end{proof}
Therefore, if $P$ is $\ell$-monotone and $\ell^{\perp}$ is perpendicular to $\ell$, then $\alpha(P)$ is $\ell^*$-monotone where $\ell^*$ is perpendicular to $\alpha(\ell^{\perp})$.  

Thus, given a simple polygon $P$ whose vertices are in general position and a fixed direction $\delta$, since the trapezoid decomposition of $P$ with respect to $\delta$ is unique, we have a $O(n)$ time polygon triangulation algorithm that is affine invariant. Using the trapezoid decomposition, the polygon can be decomposed into monotone pieces that are also affine invariant since they result from adding diagonals
adjacent to $\ell^{\perp}$-cups that appear on the boundary of the trapezoids. Monotone polygons can be triangulated in linear time given a sweep-line ordering of the vertices along the direction of monotonicity. As a consequence of Theorem~\ref{thm:sweeping} and that this ordering can be recovered in linear time, we conclude with the following.
 
\begin{corollary}
There exists an $O(n)$ affine invariant algorithm for triangulating a simple polygon $P$ with vertex set in general position
\end{corollary}
 
 \section{Conclusions}
 
 In this paper we revisited Nielson's affine invariant norm, whose unit disk represents how spread the point set is with respect to its mean.
 We also proposed affine invariant point sorting methods, which are necessary for other affine invariant geometric constructions.
Our methods heavily rely on being able to distinguish three points. To this end, we gave two methods for distinguishing such points in Section~\ref{sec:othergeom}. However, for this we had to require that the points be in a certain general position.
Otherwise, 
the point set becomes highly symmetric, which introduces a problem of indistinguishability. 
 In particular, an interesting open question is to what extent such restrictions can be removed, while still being able to distinguish rotations and reflections. 
Finally, the examples we provided of affine invariant algorithms that compute geometric object is by no means an exhaustive list, as such, we believe that finding affine invariant methods to construct other geometric objects is a promising direction for future research.

\bibliographystyle{plain}
\bibliography{extended_version.bbl}

\end{document}